\documentclass[10pt,aps,tightenlines,nofootinbib, twocolumn]{revtex4-2}
\usepackage{graphicx}

\pdfoutput=1
\usepackage{hyperref}
\hypersetup{
	colorlinks=true,
	urlcolor=blue,
	linkcolor=blue,
	citecolor=blue
}
\usepackage{url}
\usepackage{amsmath}
\usepackage{amssymb}
\usepackage{amsthm}
\usepackage{bbm}
\usepackage{mathtools}
\usepackage[normalem]{ulem}
\usepackage{dsfont}
\usepackage{mathrsfs}
\usepackage[makeroom]{cancel}
\usepackage{graphicx}
\usepackage{bm}
\usepackage{dsfont}
\usepackage[T1]{fontenc}
\usepackage[ddmmyyyy]{datetime}
\usepackage{physics}
\usepackage{xcolor}
\usepackage{thm-restate}
\usepackage{comment}

\newtheorem{thm}{Theorem}
\newtheorem{conj}{Conjecture}
\newtheorem{lem}{Lemma}

\newtheorem{defn}{Definition}

\newtheorem{prob}{Problem}

\def\fZ{\mathbb{Z}}
\def\fF{\mathbb{F}}

\def\ra{\rangle}
\def\la{\langle}
\def\id{\mathbb{I}}
\def\Tr{\text{Tr}}
\def\LOG{\operatorname{LOG}}

\newcommand{\jacobi}[2]{\qty(\frac{#1}{#2})}
\makeatletter
\def\Dated@name{}
\makeatother

\begin{document}
\allowdisplaybreaks

\title{
Product Weyl–Heisenberg covariant mutually unbiased bases\\
and extremal non-stabilizerness}

\author{Bogdan S. Damski$^1$}
\author{Rafa{\l} Bistro\'n$^{1,2}$}
\author{Diego Ponterio$^1$}
\author{Jakub Czartowski$^3$}
\author{Karol \.{Z}yczkowski$^{1,4}$}
\affiliation{$^1$Jagiellonian University, 
Faculty of Physics, Astronomy and Applied Computer Science,
{\L}ojasiewicza 11, 30-348 Krak\'ow, Poland}
\affiliation{$^2$Doctoral School of Exact and Natural Sciences, Jagiellonian University, ul. {\L}ojasiewicza 11, 30-348 Kraków, Poland}
\affiliation{$^3$School of Physics, Trinity College Dublin, Dublin 2, Ireland}
\affiliation{$^4$Center for Theoretical Physics, Polish Academy of Sciences, Warsaw}

\date{
   Sept. 8, 2026}

\begin{abstract}

Discrete structures in product Hilbert spaces are investigated.
For monopartite systems
of size $d$
one relies on the Weyl–Heisenberg group $WH(d)$, while in the case of composite Hilbert spaces with the local dimensions $d_i$ we identify designs covariant with respect to the product group, $\bigotimes_i WH(d_i)$. In analogy with magic — a quantity attaining its maximum for states fiducial with respect to $WH(d)$ — we introduce a similar quantifier of {\sl product magic}, defined with respect to the product group. The maximum of this quantity over all equimodular vectors 
yields fiducial states that generate $d = \prod_i d_i$ \textit{a priori} isoentangled mutually unbiased bases (MUBs), which, when supplemented by the identity, form their complete set. Such fiducial states are explicitly constructed in all prime-power dimensions $d = p^n$ with $p\ge 3$. The result for $p\ge 5$  extends the construction of Klappenecker and
R{\"o}tteler, whereas for $p=3$ it is mathematically distinct and is based on Galois rings. The global maximum of the quantifier of product magic for 
$d=2^3$ yields fiducial states corresponding to the  symmetric informationally complete (SIC) 
generalized measurement of Hoggar.
Our approach 
feeds into
a unifying perspective in which highly symmetric quantum designs emerge from fiducial states with extremal properties via structured group-orbit constructions.

\end{abstract}

\maketitle

\twocolumngrid

\section{Introduction}

Mutually unbiased bases (henceforth MUBs)~\cite{Wootters1989} and symmetric informationally complete positive operator-valued measurements (SIC-POVMs or SICs)~\cite{DAriano2004} are highly symmetric structures in the Hilbert space of quantum states that lead to generalized quantum measurements with distinguished properties. These particular constellations of pure quantum states have been the object of intense investigation in the last decades, both in the quantum information community~\cite{mcnulty2024mutually, DURT2010} and in the mathematical community~\cite{https://doi.org/10.48550/arxiv.2501.03970}.
Two orthonormal bases of a Hilbert space are said to be \textit{unbiased} if each vector from either basis has constant overlap (up to a phase) with any vector of the other basis or, heuristically, when they are as different as possible. This feature encodes the \textit{complementarity} of the physical observables associated with an unbiased pair, a peculiar quantum trait highlighted already by Bohr~\cite{BOHR1928}, which has repercussions in applications like state tomography~\cite{Adamson2010, FernndezPrez2011}, cryptography~\cite{BechmannPasquinucci2000, Cerf2002, Yu2008}, and entanglement detection~\cite{Spengler2012}.

Full sets of $d+1$ MUBs are known to exist for any 
prime-power dimension, $d=p^n$, 
and several constructions
\cite{Ivonovic1981,Wootters1989}
rely
on the \textit{Galois field} $\mathbb{F}_{p^n}$ associated 
with the dimension $d$.
On the other hand, 
strong numerical evidence 
suggests that 
such sets do not exist in composite  dimensions~\cite{mcnulty2024mutually}. 

Whenever a complete set of MUBs exists, it provides a complex projective 2-design~\cite{klappenecker2005mutually}, a property shared with SIC generalized measurements: 
such structures faithfully reproduce the Haar average of any degree-2 polynomial.
As a consequence,
the mean value of the purity of the reduced states for any bipartition, a quadratic indicator of entanglement, 
averaged over a full set of MUBs or SICs is known.

With the notable exception of the Hoggar lines~\cite{Hoggar1998} in $d\,=\,8$, the known instances of SICs are usually constructed by the action of the Weyl–Heisenberg (WH) group over the full space $\mathcal{H}_d$ on a so-called fiducial state~\cite{Gedik2024}. This approach, in contrast, is not often used for states making up full sets of unbiased bases. 
However, it turns out that the two constellations can be jointly studied in the context of group frames. 
The authors of a recent work~\cite{Feng2024}
find a MUB fiducial vector in $\mathcal{H}_p$, with \textit{p} prime, which they use to construct \textit{p} WH-covariant MUBs in the form of mutually unbiased complex Hadamard matrices~\cite{TZ05}: 
such a collection can always be completed to a full set of bases by the addition of the identity matrix.

Moreover,
it was shown~\cite{Feng2024}
that whenever 
a SIC or a full set of MUBs exist, the corresponding 
fiducial vectors maximize the quantifier of non-stabilizerness otherwise known as \textit{magic},
over the entire Hilbert space or over all states with equimodular components, for SICs or for MUBs respectively. Nonetheless, these results of~\cite{Feng2024} directly apply to single-partite systems, with examples provided only for prime dimensions. 

The present work aims to extend this framework to multipartite, particularly prime-power dimensional quantum systems $\mathcal{H}_p^{\otimes n}$, for structures that we call \textit{a priori} -- or \textit{explicitly} -- isoentangled. By this label we mean that the structures are generated by local operations performed on a single state, and are thus isoentangled by construction. This notion can be contrasted with one of \textit{a posteriori} isoentangled structures, for which all the states forming  a given configuration are subsequently found to share the same degree of entanglement. Examples of \textit{a posteriori} isoentanglement are the SIC of  
 Zhu, Teo, and Englert~\cite{Zhu2010isoSIC} and
 the full set of 5 isoentangled MUBs presented in~\cite{Czartowski2020}, both for two-qubit systems.  
In analogy to magic, based on the global WH group~\cite{Dai2022,Feng2022,Feng2024}, we propose a quantifier of \textit{product magic} 
derived from the product of the local WH groups, and prove that this quantity is maximized by SIC or MUB fiducial states 
with maximization performed over the whole space, for SICs, and over equimodular states, for MUBs. 

Furthermore, we construct fiducial vectors in systems of prime-power dimension $d=p^n$ with $p\geq 3$.  For $p \geq 5$ our result extends the known construction of Klappenecker and R{\"o}tteler~\cite{klappenecker2003constructions}, based on the Galois field $\fF_{p^n}$, by introducing a new field-valued parameter $a \in \fF_{p^n}$ leading to a family of inequivalent constructions. For $p = 3$ we present a novel construction based on the interplay between the Galois ring $GR(9,n)$ and its residue field $GR(9,n)/(3) \approx F_{3^n}$. Thanks to this approach, we were able to circumvent the known result of the nonexistence of Alltop sequences for fields of characteristic $3^n$~\cite{Hall2012}, which prevents a standard construction based directly on Galois fields.

Lastly, we present two  
sporadic full sets of MUBs that were found in dimensions $d=4$ and $9$ using different methods. Although the former set,  derived in the standard way from the full set of five MUBs for two qubits~\cite{Czartowski2020}, turns out to be a product WH orbit, the latter set does not share this property.

This paper is structured as follows. 
In Section~\ref{sec:setting} we present fundamental notions like SICs, MUBs and unbiased complex
Hadamard matrices, along with a brief introduction to quantum entanglement and magic quantum states. 
Section~\ref{Extrema} is devoted to the definition of the \textit{product magic} upon which we therein derive and show our results for local-WH-covariant full sets of MUBs and SICs.
Next, in a short interlude we discuss known constructions of MUBs together with their properties.
In Section~\ref{sec:fiducial} we present a construction of fiducial vectors corresponding to local WH-covariant MUBs for prime-power dimensions with $p\geq 5$, as well as the one for $p = 3$, and state a conjecture that there is no product WH-covariant full set of MUBs for more than two qubits. Additionally, we present two sporadic full sets for systems of two qubits and two qutrits, with the latter 
obtained from three fiducial states under a subgroup of the product-WH group.
Finally, in Section~\ref{sec:summary} we summarize our findings 
and discuss some standing open problems about SICs and MUBs.

\section{Setting the scene}\label{sec:setting}

In this Section we recall basic notions for our work. We start with definitions of mutually unbiased bases (MUB) and symmetric informationally complete (SIC) generalized measurements.

\begin{defn}[MUB]
    We say that the orthonormal bases $B_1 = \qty{\ket{e_i}}$ and $B_2 = \qty{\ket{f_j}}$ in $\mathcal{H}_d$ are {\sl unbiased} if the squared overlap between any two states taken from different bases is constant,
    \begin{equation}
        \abs{\ip{e_i}{f_j}}^2 = \frac{1}{d}.
    \end{equation}
    A set of orthonormal bases $\qty{B_i}$ is said to be 
    {\sl mutually unbiased} if all bases are pairwise unbiased.
\end{defn}

It is evident from the definition that a global change of basis does not affect the mutual unbiasedness of two sets of vectors; as a consequence, any unbiased pair is equivalent to a pair composed of a unitary matrix with equimodular entries -- known as a \textit{complex Hadamard matrix}
-- with the addition of the identity matrix.
In the \textit{d}-dimensional space $ \mathcal{H}_d $ there are at most $d\,+\,1$ bases that are all pairwise unbiased, the bound being saturated for all prime-power dimensions $d=p^n$, for some integer \textit{n} and prime \textit{p}. The existence of such \textit{full sets} of MUBs in composite dimensions $d\neq p^n$ is a longstanding open problem, even for the smallest case of dimension $2\cdot 3\,=\,6$, for which it was conjectured already by Zauner in 1999 that only sets of no more than 3 MUBs exist~\cite{ZAUNER2011, mcnulty2024mutually}. By the global unitary invariance of MUBs, it follows that the problem of finding a full set in dimension $d$ is equivalent to finding $d$ mutually unbiased complex Hadamard matrices of size $d$, supplemented by identity~\cite{DURT2010}.

\begin{defn}[SIC]
    A collection of $d^2$ states $\qty{\ket{\psi_i}}$ in $\mathcal{H}_d$ is said to constitute a {\sl symmetric informationally complete} (SIC) measurement 
    if the squared absolute overlap between any two states is given by
    \begin{equation}
        \abs{\ip{\psi_i}{\psi_j}}^2 = \frac{d \delta_{ij} + 1}{d+1}.
    \end{equation}
\end{defn}

SIC measurements are known to mathematicians as \textit{complex equiangular lines}~\cite{Godsil2009}. They also play a significant role in quantum foundations as a fundamental entity in the Quantum Bayesianist (or QBist) approach to quantum mechanics~\cite{Fuchs2013, DeBrota2020}, allowing the representation of a quantum state solely in terms of the probabilities of measurement outcomes. Much like for full sets of MUBs, the existence of SICs in arbitrary dimension is still an open problem which was, again, first conjectured by Zauner~\cite{ZAUNER2011}. With the notable exception of the Hoggar lines in $d\,=\,8$~\cite{Hoggar1998}, the known instances of SICs are usually constructed by the action of the \textit{Weyl–Heisenberg} (WH) group on the full space $\mathcal{H}_d$ on a so-called fiducial state~\cite{ZAUNER2011, Gedik2024}.

\begin{defn}[Weyl–Heisenberg group]
\label{wh_group_def}
    Consider the shift and clock operators $X, Z\in \mathcal{B}(\mathcal{H}_d)$ defined by their action on the basis states as
    \begin{align}
        X\ket{i} = \ket{i+1\text{ mod }d} && Z\ket{i} = \omega^i_d \ket{i}
    \end{align}
    with $\omega_d = e^{\mathrm{i} 2\pi /d}$. The Weyl–Heisenberg 
    operators $W_{kl}$ are defined as
    \begin{equation}
    \label{Wkl}
        W_{kl} = \tau_d^{kl} X^k Z^l \text{ with } \tau = - e^{\mathrm{i}\pi/d}.
    \end{equation}
\end{defn}

The Weyl–Heisenberg group naturally singles out a discrete set of pure quantum states:

\begin{defn}[Stabilizer state]
A pure quantum state in $\mathcal{H}_d$ is said to be a stabilizer if it is a simultaneous eigenstate with unit eigenvalue of a maximal Abelian subgroup of the WH group.  
\end{defn}

The definition of stabilizer states is typically extended to a convex combination of (pure) stabilizer states. 
States outside the polytope of stabilizers are generally referred to as \textit{magic states}, and their distance from said polytope is a common measure of their \textit{non-stabilizerness}, or \textit{magic}.

Finally, one can extend the WH group into its normalizer, called the Clifford group.

\begin{defn}[Clifford group]
The Clifford group $\mathcal{C}_d$ over $\mathcal{H}_d$ is the \text{normalizer} of the WH group, that is, the subset of $\mathcal{U}_d$ such that for all $C \in \mathcal{C}_d$ and for all $W_{kl}$ in the WH group, \begin{equation}
    \label{clifford}
    C W_{kl} C^{\dagger} = W_{k'l'}
\end{equation}for some $W_{k'l'}$ in the WH group.
\end{defn}

The WH group, together with its stabilizers, is a crucial object in quantum computation. Operations given from its Clifford group applied on stabilizer states can be efficiently classically simulated~\cite{Dai2022}. 
This leads to the Gottesman-Knill theorem~\cite{Gottesman97}:
only quantum states that are not \textit{stabilizers} can give rise to the superpolynomial quantum advantage. 

To quantify how different given state is from any stabilizer various quantifiers of non-stabilizerness were proposed 
\cite{Bravyi2005}.
The measure based on 
the stabilizer $2$-R{\'e}nyi entropy 
can be estimated
in an experiment performed with 
a quantum computer~\cite{Oliviero2022}.
In the present
work we extend the following
measure called {\sl magic}~\cite{Leone2022stabREntropy, Dai2022,Feng2022,Feng2024} with simple explicit definition.

\begin{defn}[Magic]
\label{m_def}
For any quantum state $\rho$ of a system $\mathcal{H}_d$, a quantifier of its magic is given by
\begin{equation}
m(\rho) = \sum_{k,l = 0}^{d-1} \abs{\mathrm{Tr}[W_{{kl}}\, \rho]},
\end{equation}
where the Weyl–Heisenberg operators $W_{kl}$ are defined in 
Eq. (\ref{Wkl}).
\end{defn}
\noindent
This quantity possesses many useful properties. It is convex, constant under Clifford operations and, for any pure state $|\phi\ra$,  its magic is bounded $m(|\phi\ra)\geq d$ with equality achieved only for stabilizer states.


Lastly, we recall some relevant facts about the entanglement of quantum states. A bipartite pure state in $\mathcal{H}=\mathcal{H}_A \otimes \mathcal{H}_B$ is said to be entangled if it cannot be written as the tensor product of a vector in $\mathcal{H}_A$ and a vector in $\mathcal{H}_B$:
\begin{equation*}
    \ket{\Psi} \neq \ket{\psi_A} \otimes \ket{\psi_B}.
\end{equation*}
The amount of entanglement of $\ket{\Psi}$ can be measured in several ways, and a common choice
is the von Neumann entropy $S(\rho_A)$ of the reduced 
density operator, $\rho_A = \text{Tr}_A [|\Psi\ra\la\Psi|]$,
which varies from $0$ (for separable states) to $\log d$
(for maximally entangled, Bell-like states).
However, for this work it is convenient to use the linear
entanglement entropy, $S_{\rm lin}(\ket{\Psi})=(1-\gamma)$,
where $\gamma$ denotes the purity of the partial trace,
\begin{equation}
    \gamma(\ket{\Psi})=\text{Tr}\,\rho_A^2\;.
\end{equation} 
Hence, the smaller the purity the larger the entanglement, 
as the extremal values of purity of the reduced state
are $1$ for separable states and $1/d$ for maximally entangled states with $d = \min(d_A, d_B)$. 

Notably, this measure of entanglement is quadratic in the state. This entails that its average value over the entire Hilbert space $\mathcal{H}$ with respect to the Haar measure, given by~\cite{Lubkin1978}
\begin{equation}
    \label{average purity}
    \ev{\Tr\rho_A^2}_{\mathcal{H}} = \ev{\Tr\rho_B^2}_{\mathcal{H}} = \frac{d_A+d_B}{d_A d_B + 1},
\end{equation}
can be computed as its arithmetic mean over a 2-design, like a SIC or a full set of $(d+1)$ MUBs. The average purity of these constellations is therefore a known, fixed value. In particular, if such constellations are constructed as entanglement-preserving orbits of some fiducial vectors, it is then straightforward to see that the latter must possess this amount of entanglement. 

However, in what follows, we construct $d$ (isoentangled) {\sl mutually unbiased Hadamard matrices}~\cite{DURT2010} 
complemented by the computational basis. Thus, the average entanglement in the states comprising these $d$ complex Hadamard matrices reads,
\begin{equation} \label{eq:isoMUHs_ent}
\begin{aligned}
    \ev{\Tr\rho_A^2}_{\text{MUHM}} &= \frac{(d+1)\ev{\Tr\rho_A^2}_{\mathcal{H}} -1}{d}\\ &= \frac{d_A + d_B -1}{ d_A d_B}. 
\end{aligned}
\end{equation}

\section{Extrema of local product magic}\label{sec:extrema}
\label{Extrema}

In this section, we study 
a notion analogous to the standard magic
recalled in Definition~\ref{m_def}
and based on the global 
Weyl–Heisenberg
group~\cite{Leone2022stabREntropy, Dai2022,Feng2022,Feng2024}.
The introduced quantity, related 
to the product WH  group,
will be called \textit{product magic}.
 We present some general properties of this measure of non-stabilizerness
 and connect its extrema over all pure quantum states and over equiangular states with SICs and MUBs generated by product WH groups, respectively.
 
 In what follows, we make extensive use of the product of local WH groups, for convenience denoted in the rest of the text by the bold notation,
\begin{equation}
\label{product_wh_operators}
    \begin{aligned}
        \vb{W}_{\vb{kl}} & = \bigotimes_i W_{k_i l_i} \\
        & = \left(\prod_i \tau_{d_i}^{k_il_i}\right)\left(\bigotimes_i X^{k_i}Z^{l_i}\right),
    \end{aligned}
\end{equation}
where $\vb{k}, \vb{l}\in \qty{0,\hdots,d-1}$ are composite indices for multipartite operators. For simplicity, we will refer to this structure as the \textit{product WH group}, in contrast to the global group over the whole Hilbert space $\bigotimes_i \mathcal{H}_{d_i}$. 
 
\begin{defn}[Product magic]
For any quantum state $\rho$ of a composite system $\bigotimes_i \mathcal{H}_{d_i}$, a quantifier of its product magic is given by
\begin{equation}
M(\rho) = \sum_{\vb{k},\vb{l} = 0}^{d-1} \abs{\mathrm{Tr}[\vb{W}_{\vb{kl}}\, \rho]},
\end{equation} where the $\vb{W}_{\vb{kl}}$ are operators belonging to the product group as per Eq. \eqref{product_wh_operators}.
\end{defn}
\noindent Note that this measure can also be understood as an extension of stabilizer purity~\cite{Iannotti2026} to multiqudit systems of prime dimension; we will continue to refer to this quantity as \textit{product magic} as a quantity related to the product WH group, as opposed to \textit{magic}, as used in~\cite{Feng2024}.

The product magic measure exhibits several useful properties of its standard counterpart presented in Definition~\ref{m_def}. 

\begin{restatable}{lem}{mprop}
Let $\rho \in \Omega_d$ be any quantum state on the product of local Hilbert spaces $\mathcal{H}_{d_i}$, with respect to which we calculate the product magic $M(\rho)$. Then the following statements are true:
\begin{enumerate}
\item $M(\rho)$ is invariant under product of Clifford operations $\vb{V} =\bigotimes_i V_i$, that is, $M(\vb{V} \rho \vb{V}^\dagger) = M(\rho)$.
\item $M(\rho)$ is convex.
\item $M(\rho) \geq 1$ and the equality is achieved only for the maximally mixed state $\rho^* = \id_d/d$.
\item For any pure state $|\psi\ra$, $M(|\psi\ra) \geq d$ and the equality is achieved for the product of stabilizer states.
\end{enumerate}
\end{restatable}

All the properties listed above are proved in Appendix~\ref{app:quantifier of product magic_properties}.

To provide motivation for the above definition, we first have to quote measures quantifying the distance between
an analyzed constellation of vectors and a given 
SIC or set of MUBs~\cite{Feng2024}.

\begin{defn}
Let $\Phi = \{|\psi_\alpha\ra\}_{\alpha = 1}^{d^2}$ be a set of $d^2$ pure states in $\mathcal{H}_d$. Then a measure of how close $\Phi$ is to a SIC is
\begin{equation}
P_{\text{SIC}} = \sum_{\alpha < \beta}\left( |\la \psi_\alpha|\psi_\beta\ra| - \sqrt{\frac{1}{d+1}}\right)^2.
\end{equation}
\end{defn}

\begin{defn}
Let $\mathcal{B} =  \{|\psi_{jk}\ra\,,\, k \in \fZ_d\}_{j = 0}^{d-1}$ be $d$ orthonormal bases in $\mathcal{H}_d$. Then a measure of how close $\mathcal{B}$ is to a set of MUBs is
\begin{equation}
P_{\text{MUB}} = \sum_{j < j’}\sum_{k,k’} \left(|\la \psi_{jk}|\psi_{j’k’}\ra| - \frac{1}{\sqrt{d}} \right)^2.
\end{equation}
\end{defn}

Either quantity is equal to $0$ if and only if the corresponding set of $d^2$ vectors defines a SIC or a set of $d$ MUBs in $\mathcal{H}_d$, respectively.
It turns out that while considering a set of vectors constituting an orbit of a single fiducial state under the product Weyl–Heisenberg group, both of these divergences can be expressed using the introduced notion of product magic.

\begin{restatable}{lem}{mSIC}
Let $\mathcal{O}{|\psi\ra} = \{|\psi_{\vb{k}\vb{l}}\ra\} = \{\vb{W}_{\vb{k}\vb{l}}|\psi\ra\}$ be the orbit of a pure state $|\psi\ra$ under the product of the Weyl–Heisenberg groups on the local subsystems. Then its divergence from a SIC is given by
\begin{equation}
P_{\text{SIC}}(\mathcal{O}{|\psi\ra}) = d^3 - d^2 + \frac{d^2}{\sqrt{d+1}}(1 - M(|\psi\ra)),
\end{equation}
where $d$ is the dimension of the global Hilbert space. Therefore, $\mathcal{O}{|\psi\ra}$ generates a SIC if and only if the above expression is zero, thus
\begin{equation}
M(|\psi\ra) = 1 + (d-1)\sqrt{d+1}.
\end{equation}
\end{restatable}

\begin{restatable}{lem}{mMUB}
Let $\mathcal{O}{|\psi\ra} = \{|\psi_{\vb{k}\vb{l}}\ra\} = \{\vb{W}_{\vb{k}\vb{l}}|\psi\ra\}$ be the orbit of a pure state $|\psi\ra$ under the product of the Weyl–Heisenberg groups on the local subsystems, which can be partitioned into $d$ orthonormal bases. Then its divergence from a set of $d$ mutually unbiased bases reads,
\begin{equation}
P_{\text{MUB}}(\mathcal{O}{|\psi\ra}) = d^3 - d^2 + d\sqrt{d}(1 - M(|\psi\ra)),
\end{equation}
where $d$ is the dimension of the global Hilbert space. Therefore, $\mathcal{O}{|\psi\ra}$ generates MUBs if and only if the above expression is zero, thus
\begin{equation}
M(|\psi\ra) = 1 + (d-1)\sqrt{d}.
\end{equation}
\end{restatable}

\begin{proof}
Proofs of both Lemmas 2 and 3 are presented in Appendix~\ref{app:quantifier of product magic_saturation}.
\end{proof}

We next demonstrate that the amount of product magic necessary for a pure state $|\psi\ra$ to generate a SIC as an orbit under products of Weyl–Heisenberg groups is maximal.

\begin{lem}\label{lem:max_magic}
The maximal value of product magic is achieved by the fiducial vector of a SIC measurement (whenever it exists),
\begin{equation*}
M(\rho)\leq M(|f_{\text{SIC}}\ra) = 1 + (d-1)\sqrt{d+1}
\end{equation*}
for all quantum states $\rho \in \Omega_d$.
\end{lem}
This fact echoes the result of~\cite{Cuffaro2024} that in Hilbert spaces of any dimension, a state saturates a certain upper bound of the standard magic if and only if it belongs to a SIC measurement covariant with respect to the global WH group.
\begin{proof}
By the convexity of product magic, it is sufficient to consider pure states. Let us take an arbitrary pure state and bound its product magic by
\begin{equation*}
\begin{aligned}
M(|\psi\ra) & = \sum_{\vb{k},\vb{l}} |\Tr[\vb{W}_{\vb{k}\vb{l}} |\psi\ra\la\psi|]| \\
& = 1 + \sqrt{\Big(\sum_{(\vb{k}, \vb{l})\neq (\vb{0},\vb{0})} \!\!\!\!\!\!|\Tr[\vb{W}_{\vb{k}\vb{l}} \op{\psi}]| \Big)^2}  \\
& \leq 1 + \sqrt{(d^2-1)\!\!\!\!\sum_{(\vb{k}, \vb{l})\neq (\vb{0},\vb{0})} \!\!\!\!\!|\Tr[\vb{W}_{\vb{k}\vb{l}} \op{\psi}]|^2},
\end{aligned}
\end{equation*}
where we extracted the identity in the zeroth term and applied the Cauchy–Schwarz inequality for the sum of the remaining terms. The sum of the squared overlaps was already calculated in \eqref{eq:squared_overlaps}, but this time we need to extract from it the zeroth term $|\Tr[\vb{W}_{\vb{0}\vb{0}}|\psi\ra\la\psi|]|^2 = |\Tr[|\psi\ra\la\psi|]|^2 = 1$. Therefore, we obtain
\begin{align*}
M(|\psi\ra) & \leq 1 + \sqrt{(d-1)(d^2 - 1)} \\
& = 1 + (d-1)\sqrt{d+1}.
\end{align*}
\end{proof}

An example of a SIC with fiducial vectors maximizing the product magic are the so-called Hoggar lines~\cite{Hoggar1998, Howard2017FaultTolerant}, given by the fiducial vector
\begin{equation*}
|\psi\ra = \frac{1}{\sqrt{6}}(1+\mathrm{i}, 0, -1, 1, -\mathrm{i}, -1, 0, 0).
\end{equation*}

To present similar result for MUBs, we first show that, in the case of an equimodular state, the orbit under the product WH group consists of a set of bases.

\begin{lem}\label{lem:basis_gen}
Let $|\theta\ra$ be an equimodular vector $|\theta\ra = \frac{1}{\sqrt{d}}\sum_{\vb{j}} e^{\mathrm{i} \theta_{\vb{j}}} |\vb{j}\ra$ with arbitrary phases $\theta_{\vb{j}} \in [0, 2 \pi)$. Then its orbit under the product Weyl–Heisenberg group can be partitioned into $d$ orthonormal bases $\{|\theta_{\vb{k}\vb{l}}\ra := \vb{W}_{\vb{k}\vb{l}}|\theta\ra\}_{\vb{l} = 0}^{d-1}$ for each $\vb{k}$.
\end{lem}

\begin{proof}
Let us rewrite the initial state using the diagonal gate $U_\theta = \sum_{\vb{j} = 0}^{d-1} e^{\mathrm{i} \theta_{\vb{j}}}|\vb{j}\ra\la\vb{j}|$ as  $|\theta\ra = U_\theta |+\ra$. Then, one can check the orthogonality by appropriately commuting diagonal unitaries
\begin{equation*}
\begin{aligned}
|\la\theta_{\vb{k}\vb{l}}|\theta_{\vb{k}\vb{l}'}\ra| &= |\la +|U_\theta^\dagger Z^{-\vb{l}} X^{-\vb{k}} X^{\vb{k}} Z^{\vb{l}'} U_\theta |+\ra| \\
& = |\la +|U_\theta^\dagger Z^{\vb{l}'-\vb{l}} U_\theta |+\ra| \\
& = |\la +| Z^{\vb{l}'-\vb{l}}  |+\ra| \\
& = \prod_i\left(\frac{1}{d_i} \sum_{j_i=0}^{d_i -1} \omega_{d_i}^{j_i(l_i' - l_i)} \right) \\
& =\delta_{\vb{l}\vb{l}'},
\end{aligned}
\end{equation*}
where we used the fact that the $|+\ra = \frac{1}{\sqrt{d}}\sum_{\vb{i}} |\vb{i}\ra$ state is a product of local $|+\ra_{d_i}$ states on all subsystems, and the operators $Z^{\vb{l}'-\vb{l}}$ are also in product form.
\end{proof}

After having established a link between orbits of equimodular vectors and orthonormal bases, we 
will show below that the degree of unbiasedness of the bases
is also related to the maximization of product magic.

\begin{lem}\label{lem:eq_mod}
The maximal value of product magic over equimodular vectors is achieved for the fiducial state $|f_{\text{MUB}}\ra\in \mathcal{H}_d$ (whenever it exists), which generates $d^2$ states forming $d$ MUBs, thus
\begin{equation*}
M(|\theta\ra)\leq M(|f_{\text{MUB}}\ra) = 1 + (d-1)\sqrt{d}
\end{equation*}
for all equimodular quantum states.

\end{lem}

\begin{proof}
By the same token as in proof of Lemma~\ref{lem:max_magic} we may express the product magic of any equimodular vector $|\theta\ra$ as
\begin{equation*}
\begin{aligned}
M(|\theta\ra) & = \sum_{\vb{k},\vb{l}} |\Tr[\vb{W}_{\vb{k}\vb{l}} |\theta\ra\la\theta|]| \\
& = 1 + \sqrt{\left(\sum_{\vb{k} = 1, \vb{l}  = 0}^{d-1} |\la\theta|\vb{W}_{\vb{k}\vb{l}} |\theta\ra| \right)^2} \\
& \leq 1 + \sqrt{d(d-1)\sum_{\vb{k} = 1, \vb{l}  = 0}^{d-1} |\la\theta|\vb{W}_{\vb{k}\vb{l}} |\theta\ra|^2 }.
\end{aligned}
\end{equation*}
However, this time for each $\vb{k}$ the set $\{\vb{W}_{\vb{k}\vb{l}}|\theta\ra\}_{\vb{l} = 0}^{d-1}$ forms an orthonormal basis, according to Lemma~\ref{lem:basis_gen}, in which we can expand the state $|\theta\ra$. Furthermore, we omitted the first, “incomplete” basis with $\vb{k} = 0$ since the coefficients over which we sum are equal to zero. It is so because $\vb{W}_{\vb{0}\vb{l}}|\theta\ra$ are vectors complementing $\vb{W}_{\vb{0}\vb{0}}|\theta\ra =|\theta\ra$ to orthonormal basis. Thus, we obtain the desired result.
\end{proof}


The general recipe for a fiducial MUB vector can be understood in the following way: One first generates an equal superposition state $\ket{+} = \frac{1}{\sqrt{d}} \sum_{\vb{j}=0}^{d-1} \ket{\vb{j}}$, and then apply an appropriately constructed diagonal $T$ gate, so that $\ket{f_{\text{MUB}}} := T\ket{+}$ maximizes product magic over all equimodular states. Single qubit example for $d=2$, presented in Fig.~\ref{fig:fiducial}, provides a geometric intuition, extending to higher-dimensional analogues.

\begin{figure}[h]
\begin{center}
\includegraphics[width=1.0\linewidth]{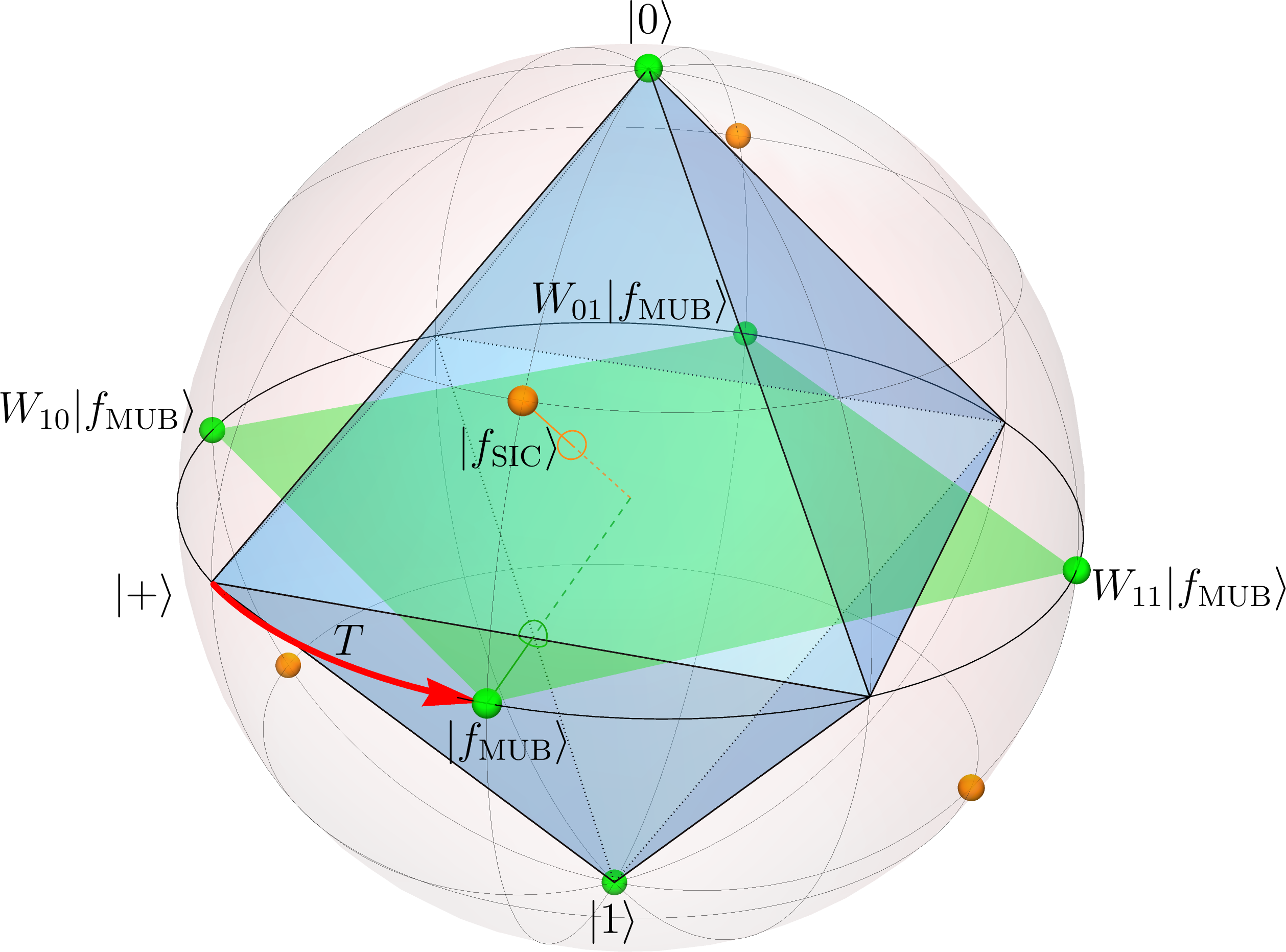} \ \
\end{center}
\caption{
\textbf{WH-covariant MUB and SIC 
in dimension $d = 2$.}
The MUB fiducial vector $|f_{\text{MUB}}\ra$ is generated from the equatorial state $|+\rangle$ with a suitable diagonal $T$ gate (rotating the state about the $Z$ axis), and the remaining equimodular states of mutually unbiased basis are generated by $W_{01} \propto Z,\,W_{10}\propto X,\, W_{11}\propto Y$ (vertices of green square). In order to obtain a full set of MUBs (green circles) one needs additionally the computational basis, which itself is WH-covariant since $\ket{0} = Z\ket{0}$ and $\ket{1} = X\ket{0} = XZ\ket{0}$.
The SIC fiducial vector $\ket{f_{\text{SIC}}}$ and its images under the WH group (orange circles) lie ``above'' the centres of octahedron faces. The stabilizer pure states are represented by vertices of an octahedron inscribed in the Bloch ball.
}
\label{fig:fiducial}
\end{figure}

Using similar methods, one can also present the counterpart of Proposition 3 from~\cite{Feng2024}, which would state that the orbit of the product of stabilizer states under the product Weyl–Heisenberg group collapses to one orthonormal basis.

Finally, we present an observation which allows us to extend the presented construction of isoentangled bases beyond equimodular vectors.
Loosely speaking, the idea is based on the fact that the above equimodular construction corresponds to a basis of stabilizer states $|\vb{i}\ra$ which are a joint eigenbasis of $\{\vb{Z}_{\vb{l}}\}_{\vb{l} = 0}^{d-1}$. However, one could choose any other stabilizer basis as a computational one. To present the statement rigorously, let us partition the product of Weyl–Heisenberg groups, with extracted identity, into $d+1$ commuting sets
\begin{equation}
\label{eq:wh_strata}
\begin{aligned}
\mathcal{W}_{\vb{j}} &= \{(\vb{W}_{1\vb{j}})^{\vb{k}},\vb{k} = 1,\cdots,d-1\}, ~ \vb{j} \in [d] \\
\mathcal{W}_{d} &= \{\vb{W}_{0\vb{k}},\vb{k} = 1,\cdots,d-1\},
\end{aligned}
\end{equation}
where the multi-index exponent $\vb{k}$ in the group $\mathcal{W}_{\vb{j}}$ should be understood as a series of exponents for each local subsystem.
A pure quantum state $|\psi\ra$ is equimodular, that is,  $|\psi\ra = |\theta\ra = U_{\theta}|+\ra$, if and only if it has zero overlap with all diagonal elements from the product of Weyl–Heisenberg groups. The implication in one direction is straightforward:
\begin{equation*}
\begin{aligned}
& \Tr[\vb{W}_{\vb{0}\vb{l}} |\theta\ra\la\theta|] = \Tr[U_\theta^\dagger \vb{W}_{\vb{0}\vb{l}} U_\theta |+\ra\la+|] \\ 
& = \Tr[\vb{W}_{\vb{0}\vb{l}}|+\ra\la+|] = 0 ,
\end{aligned}
\end{equation*}
whereas the implication in the opposite direction follows from the fact that any state can be expressed as $|\psi\ra = U_{\psi}|+\ra$ for some unitary $U_\psi$, and the operators $\vb{W}_{\vb{0}\vb{l}}$ span a basis of diagonal matrices. By expanding this observation, we can extend
Lemma~\ref{lem:eq_mod}.

\begin{lem}
Any pure state $|\psi\ra \in \mathcal{H}_d$ satisfying
\begin{equation*}
\mathrm{Tr}[\vb{W}_{\vb{k}\vb{l}} |\psi\ra\la\psi|] = 0 \text{ for all } \vb{W}_{\vb{k}\vb{l}} \in \mathcal{W}_{\vb{j}}
\end{equation*}
for some $\vb{j} = [j_1,\cdots, j_n]$,
can be mapped to an equimodular vector by some product of Clifford operations.
Thus, its measure of product magic is bounded by
\begin{equation}
M(|\psi\ra) \leq 1 + (d-1)\sqrt{d},
\end{equation}
and the above inequality is achieved if and only if $|\psi\ra$ is a fiducial vector of MUBs.
\end{lem}

\begin{proof}
The above statement follows from the fact that each of the sets $\mathcal{W}_{\vb{j}}$ \eqref{eq:wh_strata} can be mapped into $\mathcal{W}_{d}$ by some Clifford operator $C$:
\begin{equation*}
C\, \mathcal{W}_{\vb{j}}\, C^\dagger = \mathcal{W}_{d}.
\end{equation*}
where we understand $C$ as an element of the normalizer of the product WH group.
Thus, the corresponding state transformation $C|\psi\ra$ yields an equimodular state. Combining this observation with the above Lemmas and
the fact that Clifford operators 
do not change the degree of magic,
we arrive at the desired result.
\end{proof}

The above Lemma allows us to construct new families of isoentangled bases given one equimodular family by applying operations from the normalizer of the product WH group. Furthermore, since there are $d+1$ commuting sets \eqref{eq:wh_strata} we can generalize as many families of MUBs.

It is worth noting that in all cases, the maximal values of the product magic measure coincide with the maximal values of the global magic measure for corresponding scenarios.

Finally, we can show the relation between the product of magic for fiducial vectors of local MUBs
and SICs in each subspace
and the quantifier of product magic of fiducial vectors of a priori
isoentangled MUBs and SICs.

\begin{restatable}{obs}{lMcomp}
Let $\mathcal{H}_d = \bigotimes_{i = 0}^{n-1} \mathcal{H}_{d_i}$ be any product of local vector spaces,  $|f_{\text{MUB}}\ra$ and $|f_{\text{SIC}}\ra$ fiducial vectors of isoentangled MUBs and SICs in $\mathcal{H}_d$, and  $|f_{\text{MUB}, i}\ra$ and $|f_{\text{SIC}, i}\ra$ fiducial vectors of local MUBs and SICs in each subspace. Then
\begin{equation*}
\begin{aligned}
M(|f_{\text{MUB}}\ra) &> \prod_{i = 0}^{n-1} m(|f_{\text{MUB},i}\ra), \\
M(|f_{\text{SIC}}\ra) &> \prod_{i = 0}^{n-1} m(|f_{\text{SIC},i}\ra),
\end{aligned}
\end{equation*}
where for local vectors we considered the magic measure of each subsystem separately.
\end{restatable}

\begin{proof}
The claim follows from direct computation, as presented in Appendix~\ref{app:prod}.
\end{proof}

\section*{Interlude: Fiducial vectors for single party MUBs}

In this short section, we recall the results already obtained in~\cite{Feng2024} and
adopt them for our needs.
Note that the $T$ gate on systems of prime dimension $p\geq 5$ put forward therein is defined as $T = \sum_{j=0}^{p-1} \omega_p^{j^3} \op{j}$. Then, the fiducial state 
\begin{equation}
\label{f_MUBp}
\ket{f_{\text{MUB,p}}} = T\ket{+} = \sum_{j=0}^{p-1} \omega_p^{j^3} \ket{j} 
\in {\cal H}_p
\end{equation}
provides a fiducial MUB under the WH group in prime dimensions. 
This solution happens to coincide with the solution given by Alltop in~\cite{alltop1980csequences} in the work on low correlation complex sequences. 

However, the importance of this solution
to the quantum information community has gone unnoticed  
as it was found already in 1980, before
the notion of mutually unbiased bases
was established in 1989 by 
Wootters
and Fields~\cite{Wootters1989}. 
Similarly, Klappenecker and R{\"o}tteler in~\cite{klappenecker2003constructions} have not 
emphasized
the WH-covariance of this solution while providing its multipartite extension. 
The WH-covariance of Alltop MUBs in prime dimensions was elucidated 
later by Blanchfield~\cite{Blanchfield_2014} and further elucidated by Bengtsson \textit{et al.}~\cite{Bengtsson2014}.

We note below that one can extend the solution given above by putting forward a family of $T$ gates $T_a = \sum_{j=0}^{p-1} \omega_p^{aj^3} \op{j}$ with $a\in\fF_p$. As the proof is relatively simple and uses basic properties of Gaussian sums, we leave it as an exercise to the reader. However, it is worth noting that the map $\varphi:a\mapsto a^3$ over a finite field $\fF_{p}$ is an isomorphism if and only if $p \neq 1 \text{ mod }3$. Whenever $\varphi$ fails to be an isomorphism, the family of gates $T_a$ cannot be reduced to simple permutations of indices, thus giving rise to a set of inequivalent gates generating equimodular states $T_a\ket{+}$ of maximal magic from the same input. Furthermore, given the product structure of Hilbert space, even permutations of indices -- and consequently basis states -- can result in different entanglement.

Additionally, due to an irregular behaviour of cube roots for $d = p = 2, 3$ resulting in the vanishing of the quadratic terms in the inner products, the $T$ gates 
in~\cite{Feng2024}
are redefined with a common formula, $T = \sum_{i=0}^{p-1} \omega_{p^2}^{i^3}\op{i}$. Thus, the structure is not based directly on the finite field $\fF_p$ -- a fact which will turn out to be useful for the multipartite solution for $p = 3$. 

An alternative approach to generalizing the $T$ gate to the setting of qudits with prime dimension, $d = p$, based on its membership in the stabilizer group of the Clifford group, has been analysed in~\cite{Howard2012QuditT}, where the special case for $d=3$ has been introduced.



\section{Fiducial vectors for multipartite isoentangled MUBs}
\label{sec:fiducial}

In what follows we extend the concept of fiducial vectors for MUBs. Inspired by the construction put forward in~\cite{Feng2024}, we
present its extension for multipartite systems of prime-power dimension. Similarly to the case of single-party systems, the problem splits into three cases, $p\geq 5$, $p = 3$ and $p = 2$, which we tackle separately. We provide two constructions which cover all dimensions of the form $d = p^n$ for $p\geq 3$, while stating a no-go conjecture for $d=2^n$ with $n\geq 3$. \\ 
It is convenient to start with the general case $p \geq 5$.

\subsection{$p\geq 5$ case}

It turns out that the problem of the existence of a set of MUBs covariant with respect to the product group,
$WH(p)^{\otimes n}$, admits common solution for all primes $p \geq 5$.

    \begin{thm} \label{thm:f_mub}
        Consider a prime number $p\geq 5$, $n\in\mathbb{N}$, and a state
        of size $d=p^n$ of the form
        \begin{equation}\label{eq:MUB_fid}
        \ket{f_{\text{MUB}, p^n}} = \frac{1}{p^{n/2}}\sum_{j=0}^{p^n-1} \omega_p^{\tr_{\, \fF} \left[a\vb{j}^{3}\right]} \ket{j}
         \in {\cal H}_p^{\otimes n}
        \end{equation}
        with $a\neq 0$; $\vb{j}$ are elements of the Galois field $\fF_{p^n}$ and $\tr_{\, \fF}$ is the field trace with respect to the said field.
        Then, $\ket{f_{\text{MUB},p^n}}$ is a MUB fiducial state under the action of the product WH group.
    \end{thm}
    
    \begin{proof}
        The proof, based on the properties of the Galois fields, is presented in Appendix~\ref{app_proof}.
    \end{proof}

    The indices of basis states are immersed into $\fF_{p^n}$ in a natural way $\vb{j} = j_0 + j_1 x + \cdots$, where $j_0,j_1\cdots$ are indices of consecutive subsystems, and $\vb{j}$ is a polynomial associated with the element of $\fF_{p^n}$.
    We refer the reader interested in the structure of Galois fields to~\cite{Lidl_Finite_Fields, Kibler_Galois_Rings}.
    An exemplary realization of the construction for $p = 5, n = 2, a = 1$ is provided in Appendix~\ref{app:examples}. Note that for $n = 1$, the above construction reduces to the well known case of a single-partite system~\cite{Feng2024}.

    The explicit expression
(\ref{eq:MUB_fid}) for the
MUB fiducial vector in prime-power dimension $d=p^n$ constitutes
the key result of this work.
Although for higher dimensions
it is not easy to provide its geometric interpretation, one can rely on an analogy to the single qubit case, $d=2$,  presented in Fig.1.

    Inspection of the intermediate steps of the proof, together with explicit evaluation, reveals 
    that for $a=1$
    the present approach leads to the solution of Klappenecker and R{\"o}tteler~\cite{klappenecker2003constructions}
    and therefore constitutes its generalization
    for $a\neq 1$. Note that solutions with $a=1$ exhibit a highly nontrivial entanglement structure.
    For instance, in the bipartite case, $n=2$ and $p = 5$, the Schmidt coefficients admit the closed-form
    \begin{equation*}
        \vec{\lambda} = \frac{1}{10} [3+\sqrt{5}, \;2,\; 2,\; 3-\sqrt{5},\; 0].
    \end{equation*}
    Numerical exploration for $p>5$ and $n=2$ shows that depending on whether the map $\varphi:a\mapsto a^3$ is an isomorphism  in $\fF_p$ or not, we find either two or four different Schmidt coefficient structures, respectively, with all except one admitting no closed-form.

    In general, different values of $a$ yield MUBs which are inequivalent under local unitary transformations. In particular, for every $p\geq 5$ it is possible to find values of $a$ such that the structure of Schmidt coefficients is particularly simple. Let us take 
    $\tilde{a} \in \fF_{p^2}$ such that $\tr_{\fF}[\tilde{a} b] = b_0$, where $b_0$ is the free term in the polynomial representing $b\in \fF_{p^2}$. Then

    \begin{widetext}
    \begin{equation}
    \begin{aligned}
    \rho^{(1)} := \Tr_2[|f_{\text{MUB}}\ra\la f_{\text{MUB}}|] &= \frac{1}{p} \sum_{i = 0}^{p-1}  |i\ra\la i | + \frac{1}{p^{3/2}} \epsilon_p \sum_{i \neq j} \jacobi{-3 q (i - j)}{p} \omega_p^{i^3 - j^3} |i\ra\la j|, \\
    \rho^{(2)} := \Tr_1[|f_{\text{MUB}}\ra\la f_{\text{MUB}}|] &=  \frac{1}{p}\left( |0\ra\la 0 | + \sum_{i = 1}^{p-1}( |i\ra\la i| + |i\ra\la -i| )\right) \\
    & = \frac{1}{p}\left(|0\ra\la 0| + \sum_{i = 1}^{\frac{p-1}{2}} \left(|i\ra+ |-i\ra\right)\left(\la i| + \la -i|\right)\right),\\
    \end{aligned}
    \end{equation}
    \end{widetext}
    where $\epsilon_p$ and the Jacobi symbol $\left(\frac{\cdot}{p}\right)$ depend on the finite field $\fF_p$ and $|-i\ra$  is understood modulo $p$.
    
    Since the second marginal is a convex combination of $\op{0}$ and Bell states in orthogonal subspaces, it is straightforward to see that the spectrum of both marginals, and thus the Schmidt vector of the state, is
    \begin{equation*}
        \vec{\lambda} = \frac{1}{p} \bigg[\underbrace{\vphantom{\bigg[]}2, \cdots ,2 }_{\frac{p-1}{2} \text{times}}, 1, \underbrace{\vphantom{\bigg[]}0,\cdots ,0}_{\frac{p-1}{2} \text{times}}\bigg].
    \end{equation*}

    After multiple application of the Gaussian sum formula, one can show that the marginals of fiducial vector generate a MUB-like structure under the action of the Weyl–Heisenberg group.  
    Upon defining $\rho_{(kl)}^{(x)} = W_{kl}\rho^{(x)} W_{kl}$ for $x = 1,2$, their Hilbert-Schmidt inner product gives
    \begin{equation}
        \Tr[\rho_{(kl)}^{(x)}\rho_{(k'l')}^{(x)}] = 
        \begin{cases}
        \frac{1}{p} &\text{ if } k \neq k'\\
        \frac{p-1}{p^2} &\text{ if } k = k'\text{ and } l \neq l' \\
        \frac{2p-1}{p^2} &\text{ if } k = k'\text{ and } l = l'
        \end{cases}
    \end{equation}
    
    One may also leverage the fact that a unique eigenstate to the eigenvalue $1/p$ on the second marginal is $\op{0}$. Consider the global state $|f_{\text{MUB}}\ra =  \sum_{\vb{j}=0}^{p^2-1} \omega_p^{\tr_{\fF}[\tilde{a} \vb{j}^{3}]} |\vb{j}\ra$, where the indices from local subsystem are encoded into exponent as $\vb{j} = j_1 + x j_2$ and the multiplication is performed modulo an irreducible quadratic polynomial. A post-selection, or projection of a global fiducial state, with respect to the non-degenerate eigenstate of its second marginal $|0\ra$ corresponds to tossing away all terms with non-zero $j_2$; the resulting state is thus $\left(\sum_{j_1=0}^{p-1} \omega_p^{j_1^{3}} |j_1\ra\right)\otimes |0\ra$, which gives a fiducial MUB vector on the first subsystem. On the level of marginals it means that this local fiducial is a non-degenerate eigenvector of $\rho^{(1)}$.
    
    Similar results follow for all $a$ such that $\Tr(a b) = a_* b_i \text{ mod } p$ for $i\in\{0,1\}$, where $a_* \in \fF_p$ is nonzero and $b_0, b_1$ are again the coefficients of the polynomial expansion of $\vb{b}$; this results in a total of $2(p-1)$ possibilities with simplified Schmidt structure.

\subsection{$p = 3$ case}

Theorem~\ref{thm:f_mub} ceases to work for multi-qutrit systems due to 
a relatively curious fact. Heuristically, the construction \eqref{eq:MUB_fid} hinges on exponents similar to $(a+b)^3 - a^3 = 3a^2b + 3ab^2 + b^3$, which in principle are quadratic in $a$ modulo $p$, a property which is crucial for discrete Gauss sums to appear. When $p=3$, the entire expression is reduced to $(a+b)^3 - a^3 = b^3$ (due to $3 \equiv 0 \operatorname{mod} 3$), without the quadratic term required for the Gaussian sum observed for $p\geq 5$. 

This is further supported by results on nonexistence of Alltop sequences~\cite{alltop1980csequences} for fields of characteristic $3^n$ and their relation to the so-called planar functions~\cite{Hall2012}; the discussion of such objects falls beyond the scope of this manuscript, but we mention them for completeness of the discussion.

Despite the above reasons, 
knowing that a state
 based on the $9$-th root of unity,
 $\omega_9=\exp(2 \pi i /9)$,
\begin{equation*}
    \ket{f_{\text{MUB}_3}} = \frac{1}{\sqrt{3}}\mqty(1,&\omega_9,&\omega_9^8)
\end{equation*}
forms a fiducial vector for $n=1$ and $p=3$,
we extended the results for $n>1$ in two ways.

First, we find a general construction that produces fiducial vectors, also related to $\omega_9$.

\begin{thm} \label{thm:f_mub_3}
        Consider $n\in\mathbb{N}$ and a vector of the form
        \begin{equation}\label{eq:MUB_fid_3}
        \ket{f_{\text{MUB},3^n}} = \frac{1}{3^{n/2}}
        \sum_{j=0}^{3^n-1} \omega_9^{\tr\left[a \vb{j}^{3}\right]} \ket{j} \in\mathcal{H}_3^{\otimes n},
        \end{equation}
        where $a \in GR(3^2, n)$ is not a divisor of $0$, $\mathbf j$ ranges over any complete set of representatives of the cosets of $GR(3^2,n)/(3)$ and $\tr$ is the trace with respect to $GR(3^2,1) = \fZ/9\fZ$. 
        Then, $\ket{f_{\text{MUB}}}$ is a MUB fiducial state under the action of the product WH group.
\end{thm}

    \begin{proof}
        The proof, based on the relations between the Galois ring and its residue field, is presented in Appendix~\ref{app_proof3}. Note that the expression $\omega_9^{\tr(a\mathbf j^3)}$ is independent of the choice of coset representatives.
    \end{proof}
    An exemplary realization of the construction for $n = 2, a = 1$ is provided in Appendix~\ref{app:examples}.
    \smallskip

The complete theory of Galois rings is beyond the scope of this paper; below we provide a basic outline of their construction. The Galois ring $GR(p^k,n)$ can be understood as an extension of the Galois field~\cite{Krull1924, Kibler_Galois_Rings}. The ring $GR(p^k,1)$ corresponds to integers modulo $p^k$, where $p$ is some prime number: $\fZ/p^k\fZ$. The extensions $GR(p^k,n)$, are defined as a ring of polynomials $(\fZ/p^k\fZ)[x]$ modulo any monic polynomial $P(x)$ of degree $n$, which is irreducible modulo $p$, thus $GR(p^k,n) = (\fZ/p^k\fZ)[x]/P(x)$. 
Finally, the Galois ring $GR(p^k,n)$ possesses maximal ideal $(p)$, consisting of all divisors of $0$. Cosets of the ring modulo this ideal are isomorphic to residue field $GR(p^k,n)/(p) = \fF_{p^n}$. In  the considered case the indices of basis states are immersed in $GR(3^2,n)$ as $\vb{j} = j_0 + j_1 x + \cdots$, where $j_0,j_1\cdots$ are indices of consecutive qutrits.

This begs to pose the following question -- can the above construction over $GR(9,n)$ be understood as a sort of extension of Alltop functions beyond finite fields to Galois rings? Resolution of this problem falls far beyond the scope of this work and thus we leave it as an open problem.
\begin{prob}
    Can the concept of Alltop sequences be extended to notions connected to Galois rings in a manner incorporating the solution presented in Theorem~\ref{thm:f_mub_3}?
\end{prob}

We further note that the construction of Theorem~\ref{thm:f_mub_3} for $a=1$ is closely related to a generalized $T$ gate introduced in~\cite{romanowa2026blind}, used in the context of blind quantum computing.

The second result is connected with an extended numerical search in $n=2$, which resulted in a set of $9$ isoentangled Hadamard matrices generated from a triplet of fiducial states under a subgroup of the full product WH group. They are discussed and given explicitly in the Appendix~\ref{app:sporadic_set}. This solution hints at a possibility of existence of similar tri-fiducial structures in $d = 3^n$ for $n>2$ -- a problem extending beyond the scope of this manuscript, which is thus left for future investigation.

\subsection{$p=2$ case}

After demonstrating the solution for all $p>2$, we come back to the canonical example of qubits. Here we know of only two fiducial states for $n = 1, 2$, which are given by
\begin{subequations}
    \begin{align}
        \ket{f_{\text{MUB},2}} & = \frac{1}{\sqrt{2}}\mqty(1,&\omega_8), \\
        \ket{f_{\text{MUB},4}} & =  \frac{1}{2}\mqty(1,&\omega_4,&\omega_4,&\omega_4).
    \end{align}
\end{subequations}
The single qubit example is simple to understand geometrically -- one may think of a standard octahedron with vertices along $X, Y$ and $Z$ axes, rotated by $\pi/4$ around the $Z$ axis, thus resulting in all equatorial vertices related by $\pi$ rotations around either $X$, $Y$ or both $X$ and $Y$ axes, see Fig.~\ref{fig:fiducial}. This is precisely reflected in the fiducial structure behind $\ket{f_{\text{MUB},2}}$.

It is hard to draw a similar intuition from $\ket{f_{\text{MUB},4}}$, with the only hint we know of being its relation to the earlier fully isoentangled solution given in~\cite{Czartowski2020}. One may take a complete set $\qty{B_0,\hdots,B_4}$ and rotate it so that the first basis is the identity, and the remaining four are given as $H_i = B_0^\dagger B_i$. From this one obtains the set of product WH-covariant Hadamard matrices generated by the fiducial vector $\ket{f_{\text{MUB},4}}$.

Despite the above examples, exhaustive numerical search in dimension $d = 2^3 = 8$ by checking all candidates up to WH symmetries has shown that there is no fiducial vector $\ket{f_{\text{MUB},8}}$ such that $\ip{i}{f_{\text{MUB},8}} = \frac{1}{2\sqrt{2}}\omega_8^{m_i},\,m_i\in\mathbb{Z}_8$ for all computational basis states $\ket{i}$. This points us to the following conjecture.
\begin{conj}
\label{conj:no_qubit_fid}
    There exists no fiducial MUB vector under the product WH group $\ket{\psi}$ of the form
    \begin{equation}
        \forall i \ip{\psi}{i} = \frac{1}{2\sqrt{2}}\omega_8^{m_i},\, m_i \in \mathbb{Z}_8
    \end{equation}
    in dimension $d = 2^n$ for $n > 2$.
\end{conj}

It is conceivable that there exist fiducial vectors based on higher roots of unity, but we would go as far as to suggest that a stronger form of Conjecture~\ref{conj:no_qubit_fid} may hold -- that no multiqubit fiducial MUB vectors exist under the product WH group. 
Note an analogy to the earlier result by Godsil and Roy, which proved the non-existence of product WH-covariant SICs for $d = 2^n$ unless \mbox{$n = 1,\, 3$}~\cite[Lemma 3.1]{Godsil2009}. 
This observation provides further arguments in favour 
of the above conjecture.
 %

\section{Concluding Remarks}
\label{sec:summary}

\begin{table*}[t]
    \centering
    \scalebox{.8}{\begin{tabular}{lll}
    \hline
         & Complete MUBs & Dimension $d$ \\ \hline \hline
        Low periodic correlation sequences & $\qty{\ket{j}:j\in\fZ_{d}}$ & any prime $d = p\geq 5$ (Alltop~\cite{alltop1980csequences})  \\
         & $\qty{\frac{1}{\sqrt{d}}\sum_{j=0}^{d-1} \omega_p^{(j+k)^3 + l(j+k)}\ket{j}: l\in\fZ_d}, k\in\fZ_d$ &  \\ \hline
        Alltop extension to finite fields & $\qty{\ket{j}:j\in\fF_{d}}$ & any prime-power $d = p^n$, $p\geq 5$  \\
         & $\qty{\frac{1}{\sqrt{d}}\sum_{j=0}^{d-1} \omega_p^{\tr_\fF[(j+k)^3 + l(j+k)]}\ket{j}: l\in\fF_d}, k\in\fF_d$ & (Klappenecker and R{\"o}tteler~\cite{klappenecker2003constructions}  \\ \hline
        Via orthogonal unitary matrices & $\qty{\ket{j}:j\in\fF_{d}}$ & any prime $d = p \geq 3$ (Ivonovic~\cite{Ivonovic1981})\\
         & $\qty{\frac{1}{\sqrt{d}}\sum_{j=0}^{d-1} \omega_p^{\tr_\fF[kj^2 + lj]}\ket{j}: l\in\fF_d}, k\in\fF_d$ & any prime-power $d = p^n$, $p\geq 3$ (Wootters and Fields~\cite{Wootters1989})\\ \hline
         Eigenstates of Pauli group & $\qty{\ket{j}:j\in\fZ_{d}}$ & any prime $d = p$ (Bandyopadhyay \textit{et al.}~\cite{Bandyopadhyay2002}) \\
         & $\qty{\text{Eigenstates of }XZ^k}, k\in\mathbb{Z}_p$ & (prime-power formulae after modification) \\ \hline
         From fiducial states & $\qty{\ket{j}:j\in\fZ_{d}}$ & any prime $d = p\geq 5$ \\
         & $\qty{W_{kl}\ket{f_{\text{MUB}}}, k\in\fZ_d}, l\in\fZ_d$ &  (Feng and Luo~\cite{Feng2024}, Alltop~\cite{alltop1980csequences}) \\
         & with $\ket{f_{\text{MUB}}} = \frac{1}{\sqrt{d}}\sum_{i=0}^{d-1} \omega_d^{i^3}\ket{i}$ & (special cases for p = 2, 3) \\ \cline{2-3}
         & $\qty{\ket{j}:j\in\fF_{d}}$ & any prime-power $d = p^n, p\geq 5$ \\
         & $\qty{\vb{W_{kl}}\ket{f_{\text{MUB}}}, k\in\fF_d}, l\in\fF_d$ &  (This work, Theorem~\ref{thm:f_mub})\\
         & with $\ket{f_{\text{MUB}}} = \frac{1}{\sqrt{d}}\sum_{i=0}^{d-1} \omega_p^{\tr_\fF(a i^3)}\ket{i}$ & (Klappenecker and R{\"o}tteler~\cite{klappenecker2003constructions} for $a=1$)  \\ \cline{2-3}
         & $\qty{\ket{j}:j\in \fF_{3^n}}$ & any prime-power $d = 3^n$ \\
         & $\qty{\vb{W_{kl}}\ket{f_{\text{MUB}}}, k\in \fF_{3^n}}, l\in \fF_{3^n}$ &  (This work, Theorem~\ref{thm:f_mub_3}) \\
         & with $\ket{f_{\text{MUB}}} = \frac{1}{\sqrt{d}}\sum_{i=0}^{d-1} \omega_9^{\tr_{GR}(a i^3)}\ket{i}$ &   \\ \cline{2-3}
    \end{tabular}}
    \caption{Selection of known MUB constructions with two last rows presenting novel contribution from this work. All the constructions above, except the last one, are based on the finite field $\fF_{p^n}$.
    Note that in the last case, $d = 3^n$, the indices $j,k,l$ are marked as elements of $\fF_{3^n}$. However, they run over elements of Galois ring $GR(9,n)$, each belonging to different coset $GR(9,n)/(3)$, and are identified with elements of residue field $\fF_{3^n}$ only during the calculation of overlaps.}
    \label{tab:MUB_constructions}
\end{table*}

In this work, we have investigated the properties of product magic as the local counterpart of magic and analyzed its role in the context of product Weyl–Heisenberg (WH) covariant structures. We have shown that fiducial states of symmetric informationally complete (SIC) measurements and mutually unbiased bases (MUBs), when covariant with respect to the product WH group and provided that such structures exist, are necessarily maximizers of product magic among all states and equatorial states, respectively. This extends the features previously demonstrated for the global WH group~\cite{Feng2024} to the multipartite setting.

We constructed fiducial states in ${\cal H}_d$,
which lead to $d$ a priori isoentangled MUBs in prime-power dimension, $d=p^n$, with $p \geq 3$.
In the qubit case, $p=2$,  
explicit examples of MUB fiducial states are provided for $d=2$ and $d=2^2$, while for $n>2$ parties
the non-existence of an analogous construction is conjectured. 
In this sense, our approach complements the well-established method of constructing SICs from fiducial states, by highlighting structural constraints specific to product-group covariance. 
It may also be interesting to explore the relation of the generalized $T$ gate to the stabilizer of the multipartite Clifford group or its appropriate modification, following lines 
presented in~\cite{Howard2012QuditT, Cui2017diagonalClif}.

The presented results feed into a variety of different constructions for MUBs, a selection of which we present in Table~\ref{tab:MUB_constructions} for context.
Our construction of fiducial MUBs for $p\geq 5$ extends the low correlation sequences by Alltop~\cite{alltop1980csequences} for single systems, as well as the multipartite solution given by Klappenecker and R{\"o}tteler~\cite{klappenecker2003constructions}, jointly providing a family of solutions indexed by $a\in\fF_{p^n}$. For the bipartite case we demonstrated that the entanglement structure of the resulting states changes depending on the parameter $a$; furthermore, numerical exploration suggests that, depending on the properties of the third power in $\fF_p$, different entanglement structures emerge.

The results of this work naturally feed into the ongoing line of research concerned with the entanglement properties of mutually unbiased bases~\cite{Romero2005, Wiesniak_2011}. Owing to the local nature of the product WH group, 
the construction presented leads to a significant simplification compared to previous approaches: all nonlocal features are confined to the preparation of the fiducial state, while the remaining basis states are generated by purely local operations, assuming the computational basis as the natural reference.

Perhaps most importantly and surprisingly, our multipartite construction for $p=3$ goes beyond the known construction for MUBs and circumvents known no-go theorems for Alltop functions~\cite{alltop1980csequences}
over finite fields by replacing them with Galois rings. This result suggests a deeper investigation into the properties of the underlying mathematical objects, which may provide equivalent of Alltop functions extended to the novel setting of Galois rings.

The existence of a single known example of a product WH-covariant SIC measurement, namely the Hoggar lines, further motivates the search for similar structures in other multipartite systems, or alternatively, for rigorous no-go results excluding their existence beyond this exceptional case.

Furthermore, due to possible applications in quantum information theory,
there is currently considerable interest in complex Hadamard
matrices with special structure~\cite{BRZ24}.
From this perspective, the presented constructions constitute
sets of $d$ mutually unbiased complex Hadamard matrices
of order $d=p^n$ for $p\ge 3$,
which belong to the Butson class $BH(d,p)$ for $p \geq 5$, and $BH(d,9)$ for $p = 3$ -- see~\cite{TZ05}.
As shown in Appendix~\ref{app:sporadic_set},
the sporadic solution for $d=3^2$,
generated by a subgroup of the product WH group,
belongs to a narrower Butson class $BH(9,3)$,
as all entries are integer powers of $\omega_3$,
the third root of unity.

As an aside, the present work reinforces a unifying structural theme: the emergence of highly symmetric measurement and design structures from remarkably economical fiducial data. Symmetric informationally complete measurements arise as group orbits of a single fiducial vector. In all odd prime-power dimensions, a complete set of $d+1$ mutually unbiased bases (MUBs) can likewise be generated from just two fiducial vectors, including a single state from the computational basis. The sporadic construction in $d = 9$ described in Appendix~\ref{app:sporadic_set} provides an intermediate example, built from three fiducial vectors for a subgroup of the product WH group. At the opposite extreme, it is known that even a single orthonormal basis can generate a complex projective 2-design under a 
$U(1)$ action~\cite{avella2025cyclic}.

Taken together -- and in light of the results of~\cite{Feng2024} -- these examples point toward a broader organizing principle: highly structured quantum designs often admit descriptions in terms of minimal generating data and their associated group frames. This perspective motivates a systematic investigation of fiducial structures and the extent to which symmetry and design properties can be derived from compressed, orbit-based constructions.
Finally, while the present work does not address the longstanding problem of the existence of complete sets of MUBs in composite dimensions, it is well known that, should such sets exist, they cannot consist solely of maximally entangled and product states. To the best of our knowledge, the problem of constructing almost-isoentangled MUBs in dimension six, or proving a corresponding no-go theorem, remains open.

\medskip 
\begin{acknowledgments}

It is a pleasure to thank Ingemar Bengtsson, Wojciech Bruzda, 
Lorenzo Campos Venuti, Markus Grassl, David Gross,
 Philipp Hauke,  Mark Howard, Daniele Iannotti
 and Mih{\'a}ly Weiner
for useful comments and valuable interactions. 
The authors acknowledge funding by the European Union under 
ERC Advanced Grant TAtypic, Project No. 101142236.
J. Cz. acknowledges support from Taighde
Éireann – Research Ireland under Grant number IRCLA/2022/3922. 
R.B. acknowledges support by the National Science Center, Poland, under the contract number 2023/50/E/ST2/00472.

\end{acknowledgments}

\appendix
\onecolumngrid

\section{General features and maxima of product magic}

In this Appendix we present technical proofs and calculations from Section~\ref{sec:extrema}.

\subsection{Properties of product magic}
\label{app:quantifier of product magic_properties}

Let us start by presenting the properties of the non-stabilizerness monotone under consideration.

\mprop*

\begin{proof}
To prove the first statement we note that in each subsystem the conjugation by Clifford operation on elements from the Weyl–Heisenberg group is a bijection~\cite{Dai2022}, thus we have
    \begin{equation*}
    \begin{aligned}
        M(\vb{V}\rho \vb{V}^\dagger) &= \sum_{\vb{k},\vb{l}} |\Tr[\vb{V}\rho\vb{V}^\dagger \vb{W}_{\vb{k}\vb{l}}]| = \sum_{\vb{k},\vb{l}} |\Tr[\rho\vb{V}^\dagger \vb{W}_{\vb{k}\vb{l}}\vb{V}]| \\
        &= \sum_{\vb{k}',\vb{l}'} |\Tr[\rho \vb{W}_{\vb{k}'\vb{l}'}]| = M(\rho),
    \end{aligned}
    \end{equation*}
    where we defined $\vb{W}_{\vb{k}'\vb{l}'} = \vb{V}^\dagger \vb{W}_{\vb{k}\vb{l}}\vb{V}$.

    The second statement follows from direct calculations
    \begin{equation*}
    \begin{aligned}
    &M(p \rho + (1-p) \sigma) = \sum_{\vb{k},\vb{l}} | p\Tr[\rho \vb{W}_{\vb{k}\vb{l}}] + (1-p)\Tr[\sigma \vb{W}_{\vb{k}\vb{l}}]| \\
    & \leq p \sum_{\vb{k},\vb{l}} |\Tr[\rho \vb{W}_{\vb{k}\vb{l}}] | + (1-p) \sum_{\vb{k},\vb{l}}|\Tr[\sigma \vb{W}_{\vb{k}\vb{l}}]|  =p M(\rho) + (1-p) M(\sigma).
    \end{aligned}
    \end{equation*}

    To show the third point we notice that 
    \begin{equation*}
        M(\rho) = \sum_{\vb{k},\vb{l}}|\Tr[\vb{W}_{\vb{k}\vb{l}} \rho]| \geq \Tr[\vb{W}_{\vb{0}\vb{0}}\rho] = \Tr[\rho] = 1,
    \end{equation*}
    where the inequality becomes an equality if and only if all the omitted terms vanish $\Tr[\vb{W}_{\vb{k}\vb{l}}\rho] = 0$. However, since the Weyl–Heisenberg group on the Hilbert space of size $d_i$ provides an orthonormal basis for matrices of size $d_i\times d_i$, their product provides the basis for matrices of size $d\times d$, so this condition implies that $\rho$ is proportional to the identity.

    To show the lower bound in the fourth point, we once again use the fact that the Weyl–Heisenberg group on the Hilbert space of size $d_i$ provides an orthonormal basis for matrices of size $d_i\times d_i$: $\frac{1}{d} \sum_{k_i,l_i} |D_{k_il_i}\ra\la D_{k_i l_i}| = \id_{d_i^2}$, where the vectorization of any operator $X$ is defined as $X|\Psi_+\ra$ with $|\Psi_+\ra = \frac{1}{\sqrt{d}}\sum_i|ii\ra$ being the Bell state.
    By using this property on each subsystem we obtain
    \begin{equation}
    \label{eq:in_square_overlaps}
    \begin{aligned} 
     M(|\psi\ra) &= \sum_{\vb{k},\vb{l}} |\Tr[\vb{W}_{\vb{k}\vb{l}}|\psi\ra\la \psi|]| \geq \sum_{\vb{k},\vb{l}} \Tr[\vb{W}_{\vb{k}\vb{l}}|\psi\ra\la \psi|]\; \Tr[|\psi\ra\la \psi|\vb{W}_{\vb{k}\vb{l}}^\dagger] \\
     &=  \sum_{\vb{k},\vb{l}} \la \psi\otimes\overline{\psi}| \vb{W}_{\vb{k}\vb{l}}\ra\la \vb{W}_{\vb{k}\vb{l}}| \psi\otimes\overline{\psi}\ra   =   d\la \psi\otimes\overline{\psi}|\bigotimes_i \id_{d_i}| \psi\otimes\overline{\psi}\ra=  d \la\psi|\psi\ra^2 = d,
    \end{aligned}
    \end{equation}
    where we used the fact that each term $|\Tr[\vb{W}_{\vb{k}\vb{l}}|\psi\ra\la \psi|]|$ must be smaller than $1$, since the state $|\psi\ra$ is normalized, so these terms are smaller than their squares.
    Therefore, the equality in the above calculations is possible if and only if $d$ of these terms are equal $1$ and all the remaining ones are equal to zero. However, since $\vb{W}_{\vb{k}\vb{l}}$ is a product of elements from Weyl–Heisenberg group, it is simple to notice that a product of stabilizer states satisfies this condition.
\end{proof}

\subsection{Relation between distances and product magic}\label{app:quantifier of product magic_saturation}

Next we connect the maxima of product magic with SICs and MUBs.

\subsubsection{SIC}
\mSIC*

\begin{proof}
    Let us start with noticing that the product Weyl–Heisenberg group has exactly $d^2$ elements, so this is the number of vectors from the considered set $m = d^2$.
    By explicit calculations we can rewrite the expression for $P_{\text{SIC}}$ as
    \begin{equation*}
    \begin{aligned}
    P_{\text{SIC}}(\mathcal{O}_{|\psi\ra})  = & \frac{1}{2}\left( \sum_{\vb{k},\vb{l},\vb{k}',\vb{l'}}\left( |\la \psi_{\vb{k} \vb{l}}|\psi_{\vb{k}' \vb{l}'}\ra| - \sqrt{\frac{d^2 - d}{d(d^2-1)}}\right)^2  -  
    \sum_{\vb{k},\vb{l}}\left( |\la \psi_{\vb{k} \vb{l}}|\psi_{\vb{k} \vb{l}}\ra| - \sqrt{\frac{d^2 - d}{d(d^2-1)}}\right)^2 \right) \\
    = & \frac{1}{2} \sum_{\vb{k},\vb{l},\vb{k}',\vb{l'}}|\la \psi_{\vb{k} \vb{l}}|\psi_{\vb{k}' \vb{l}'}\ra|^2  -  \sqrt{\frac{d^2 - d}{d(d^2-1)}}\sum_{\vb{k},\vb{l},\vb{k}',\vb{l'}}|\la \psi_{\vb{k} \vb{l}}|\psi_{\vb{k}' \vb{l}'}\ra|
    + \frac{(d^2)^2(d^2 - d)}{2 d (d^2 - 1)} \\ & -\frac{d^2}{2} \left( 1 - \sqrt{\frac{d^2 - d}{d(d^2-1)}}\right)^2.
    \end{aligned}
    \end{equation*}
    The sum of the overlaps squared can be calculated by using the same methods as in \eqref{eq:in_square_overlaps}:
    \begin{equation}
    \label{eq:squared_overlaps}
    \begin{aligned}
     \sum_{\vb{k},\vb{l},\vb{k}',\vb{l'}}|\la \psi_{\vb{k} \vb{l}}|\psi_{\vb{k}' \vb{l}'}\ra|^2 &= d^2 \sum_{\vb{k}'',\vb{l''}} |\la \psi|\psi_{\vb{k}''\vb{l}''}\ra|^2  
     = d^2 \sum_{\vb{k}'',\vb{l''}} \Tr[\vb{W}_{\vb{k}''\vb{l}''}|\psi\ra\la \psi|]\Tr[|\psi\ra\la \psi|\vb{W}_{\vb{k}''\vb{l}''}^\dagger] \\
     &=  d^2  \sum_{\vb{k}'',\vb{l''}} \la \psi\otimes\overline{\psi}| \vb{W}_{\vb{k}''\vb{l}''}\ra\la \vb{W}_{\vb{k}''\vb{l}''}| \psi\otimes\overline{\psi}\ra 
     = d^2  \la \psi\otimes\overline{\psi}|\bigotimes_i \id_{d_i}| \psi\otimes\overline{\psi}\ra \\
     & =  d^3 \la\psi|\psi\ra^2 = d^3,
     \end{aligned}
    \end{equation}
    where in the first step we change the summation argument $\vb{W}_{\vb{k}\vb{l}}^{-1}\vb{W}_{\vb{k}'\vb{l}'} = \vb{W}_{\vb{k}''\vb{l}''}$, and utilize an identity resolution for each local subspace. 
    
    Finally, the sum of the overlaps translates into our measure of product magic
  \begin{equation}
    \label{eq:lin_overlaps}
    \sum_{\vb{k},\vb{l},\vb{k}',\vb{l'}}|\la \psi_{\vb{k} \vb{l}}|\psi_{\vb{k}' \vb{l}'}\ra| = d^2 \sum_{\vb{k}'',\vb{l}''} |\la \psi_{\vb{k}'' \vb{l}''}|\psi\ra| = d^2 \sum_{\vb{k}'',\vb{l}''} \Tr[\vb{W}_{\vb{k}''\vb{l}''} |\psi\ra\la\psi|] = d^2 M(|\psi\ra).
    \end{equation}
    Thus, after straightforward calculations, we obtain the result.
\end{proof}

\subsubsection{MUBs}

\mMUB*
\begin{proof}
    Once again the cardinality of $\mathcal{O}_{|\psi\ra}$ is $d^2$, thus we consider candidates for $d$ mutually unbiased bases. To clarify the calculation we assume states from the same basis to have the same multiindex $\vb{k}$ (this does not, however, influence the results).
    By direct calculation of the divergence from MUBs we get
    \begin{equation*}
    \begin{aligned}
    P_{MUB}(\mathcal{O}_{|\psi\ra}) & = \frac{1}{2}\left(\sum_{\vb{k},\vb{l},\vb{k}',\vb{l'}} \left(|\la \psi_{\vb{k} \vb{l}}|\psi_{\vb{k}' \vb{l}'}\ra| - \frac{1}{\sqrt{d}} \right)^2
    - \sum_{\vb{k},\vb{l},\vb{l'}} \left(|\la \psi_{\vb{k} \vb{l}}|\psi_{\vb{k} \vb{l}'}\ra| - \frac{1}{\sqrt{d}} \right)^2 \right) = \\
    & = \frac{1}{2} \sum_{\vb{k},\vb{l},\vb{k}',\vb{l'}} |\la \psi_{\vb{k} \vb{l}}|\psi_{\vb{k}' \vb{l}'}\ra|^2 - \frac{1}{\sqrt{d}} \sum_{\vb{k},\vb{l},\vb{k}',\vb{l'}} |\la \psi_{\vb{k} \vb{l}}|\psi_{\vb{k}' \vb{l}'}\ra| + \frac{d^4}{2d} - d(d - \sqrt{d}),
    \end{aligned}
    \end{equation*}
    where for the second term we used the orthogonality within the same basis. The sums over the quadratic and linear overlaps can be performed in the same way as previously in \eqref{eq:squared_overlaps} and \eqref{eq:lin_overlaps}, thus the claim follows from straightforward calculations.
    \end{proof}

    \subsection{Total fiducial states have more product magic than the product of all local ones}\label{app:prod}

    In this subsection we compare the product of local magic with product magic. 

    \lMcomp*

    \begin{proof}
    We present the claim for fiducial vectors of MUBs, and omit the proof for SICs since it follows from exactly the same arguments, however with a more convoluted expression. 
    Let us show the statement by induction, starting from two-subsystem case with local dimensions $d_1,d_2$, so the global dimension is $d = d_1d_2$. Then
    \begin{equation*}
    \begin{aligned}
        M(|f_{\text{MUB}}\ra) - m(|f_{\text{MUB},1}\ra)m(|f_{\text{MUB},2}\ra) &= 1 + (d_1 d_2-1)\sqrt{d_1 d_2} - \left(1 + (d_1-1)\sqrt{d_1}\right)\left(1 + (d_2-1)\sqrt{d_2}\right) \\
        & = \sqrt{d_1 d_2}(d_1 + d_2 - 2) - (d_1 - 1)\sqrt{d_1} - (d_2 - 1)\sqrt{d_2},
    \end{aligned}
    \end{equation*}   
    If we substitute $x = \sqrt{d_1}$, $y = \sqrt{d_2}$, then the above expression can be rewritten as
    \begin{equation*}
    \begin{aligned}
    M(|f_{\text{MUB}}\ra) - m(|f_{\text{MUB},1}\ra)m(|f_{\text{MUB},2}\ra) &= xy(x^2 + y^2 - 2) - x(x^2 - 1) - y(y^2 - 1) \\
    &= x^3 y + x y^3 - 2xy - x^3 + x - y^3 + y \\
    &= x^3(y - 1) - x(y - 1) + y^3(x - 1) - y(x - 1) \\
    &= (y - 1)x(x^2 - 1) + (x - 1)y(y^2 - 1) \\
    &= (x - 1)(y - 1) \left[ x(x + 1) + y(y + 1) \right] \\
    &= (\sqrt{d_1} - 1)(\sqrt{d_2} - 1) \left( d_1 + \sqrt{d_1} + d_2 + \sqrt{d_2} \right) > 0
    \end{aligned}
    \end{equation*}
    which is obviously positive for $d_1,d_2 \geq 2$. 
    To show the claim for $m$-partite subsystem we may divide it into first party and all remaining ones and apply the above calculation to obtain
    \begin{equation*}
        M(|f_{\text{MUB}}\ra) > m(|f_{\text{MUB},0}\ra)M(|f_{\text{MUB},rest}\ra) > \prod_i m(|f_{\text{MUB},i}\ra),
    \end{equation*}
    where $M(|f_{\text{MUB},rest}\ra)$ is the magic of a fiducial vector for isoentangled MUBs on subsystems $(1,\cdots,n-1)$.
    \end{proof}

\section{Analytical formulas for $d$ a priori isoentangled MUBs in dimension $d = p^n$}\label{app_proof}

A fiducial vector under the product of multiple Weyl–Heisenberg groups $\mathcal{WH}^{\otimes n}$ 
for any prime $p \ge 5$
can be rewritten as
\begin{equation}
\label{eq:a1}
    \ket{f_{\text{MUB},p^n}} = T_{p^n}(a)\ket{+},
\end{equation}
 where the superposition state reads
$\ket{+} = p^{-{n/2}}\sum_{i=0}^{p^n-1} \ket{i}$ and 
the family of multipartite $T$ gates is defined as
\begin{equation}
\label{eq:a2}
    T_{p^n}(a) =  \sum_{j=0}^{p^n-1} \omega_p^{\tr_\fF(a j^{ 3})}\op{j},
\end{equation}
where $\omega_p = \exp(2\pi i/p)$ and the multiplication and powers in the exponent are performed with respect to the finite field $\fF_{p^n}$. We additionally take $a\neq 0$, since $T_{p^n}(0) = \mathbb{I}$.

\begin{proof}

    We consider overlaps of two states form the orbit
    \begin{equation}
        \abs{\ev{\vb{W}_{\vb{ij}}^\dagger \vb{W}_{\vb{kl}}}{f_{\text{MUB}}}} = \abs{\ev{\vb{X}^{\vb{k-i}}\vb{Z}^{\vb{l-j}}}{f_{\text{MUB}}}} = \abs{\ev{\vb{X}^{\vb{b}}\vb{Z}^{\vb{c}}}{f_{\text{MUB}}}},
    \end{equation}
    where we set $\vb{b} = \vb{k-i}$ and $\vb{c} = \vb{l-j}$.
    The identity is due to the fact that the commutation between $X$ and $Z$ introduces only a global phase, which is irrelevant to the absolute value. With this, we proceed to evaluate the absolute value

    \begin{subequations}
        \begin{align}
            \abs{\ev{\vb{X}^{\vb{b}}\vb{Z}^{\vb{c}}}{f_{\text{MUB}}}} & = \frac{1}{p^n}\abs{\sum_{i,j=0}^{p^n-1}\omega^{\tr_\fF(ai^3) - \tr_\fF(aj^3)}\omega^{\sum_\nu c_\nu i_\nu} \ip{j-b}{i}} \\ 
            & = \frac{1}{p^n}\abs{\sum_{i=0}^{p^n-1}\omega^{\tr_\fF (a[i^3-(i+b)^3] +c_* i)}} \\
            & = \frac{1}{p^n}\abs{\sum_{i=0}^{p^n-1}\omega^{-\tr_{\fF}[3a(i^2b + i (b^2 - \frac{c_*}{3a}))]}} \label{step:trace_element}\\
            & = \frac{1-\delta_{b0}}{p^n}\abs{\sum_{j=0}^{p^n-1}\omega^{-\tr_{\fF}\qty[3abj^2]}} + \frac{\delta_{b0}}{p^n}\abs{\sum_{i=0}^{p^n-1}\omega^{\tr_\fF(c_* i)}}\label{step:square_complete_relabel}\\
            & = \frac{1-\delta_{b0}}{p^n}\prod_{k=1}^n\abs{ \sum_{j_k=0}^{p-1}\omega^{\lambda_k j_k^2}} + \delta_{b0}\delta_{c0} = (1-\delta_{k i})p^{-n/2}  + \delta_{ki}\delta_{lj},\label{step:nondegenerate_diagonalization}
        \end{align}
    \end{subequations}
    where in \eqref{step:trace_element} we have used the uniqueness of the field trace $\tr_\fF: \fF_{p^n}\mapsto\fF_p$ as the linear map going to the base field, together with the existence of a unique element $c_*$ such that $\sum_\nu c_\nu i_\nu = \Tr(c_* i)$~\cite{Lidl_Finite_Fields}. In \eqref{step:square_complete_relabel} we introduced the relabelling $j = i -\frac{3ab^2 - c_*}{6ab}$, which induces a permutation of the elements of $\fF_{p^n}$ and introduces a phase factor irrelevant under the absolute value; for $b=0$ the relabelling cannot be introduced and the expression reduces to the sum of the terms with exponents linear in $i$, thus reducing this term to $\delta_{c0}$. Finally, in \eqref{step:nondegenerate_diagonalization} we utilised the fact that any non-degenerate quadratic form over $\fF_{p^n}$ as a vector space over $\fF_p$ can be diagonalised with $\lambda_k \neq 0$, thus allowing us to rewrite the expression as a product of $n$ discrete Gaussian sums, each contributing a factor of $\sqrt{p}$, which yields the final result.
\end{proof}

\section{Analytical formulas for $d$ isoentangled MUBs in dimension $d = 3^n$}\label{app_proof3}

Going back to the proof presented for $p\geq 5$, the point of failure for $p = 3$ becomes apparent already at step \eqref{step:trace_element}, and merits elucidation. 
In the field $\fF_{3^n}$, the multiplication by $3$  means $3a = a + a + a = 0$ with addition being shown explicitly as addition with respect to $\fF_{3^n}$ which vanishes identically due to the vector-like component-wise addition modulo $3$ operation involved in it.
For this reason, the $p=3$ case lacks the quadratic term, which would have produced the mutual unbiasedness between elements from different bases.

To tackle this difficulty, we consider instead the  polynomial extension of a commutative ring $(\fZ/9\fZ)[x]$ modulo $n$ degree monic polynomial $P(x)$ which would be irreducible modulo $3$, also known as  Galois ring $GR(3^2,n) := (\fZ/9\fZ)[x]/P(x)$~\cite{Krull1924, Kibler_Galois_Rings}. In this module we use the trace with respect to primary Galois ring, $\tr_{GR}[a]: GR(3^2,n)\to GR(3^2,1)$, in a standard way as a trace of multiplication by $a$ in any basis. Then the general formula for a fiducial vector in $\mathcal{H}_3^{\otimes n}$ is given by

\begin{equation}\label{eq:fid9}
    |f_{\text{MUB},3^n}\ra = \frac{1}{3^{n/2}} \sum_{j=0}^{3^n-1} \omega_9^{\tr_{GR}\left[a  j^{ 3}\right]} \ket{j},
\end{equation}
where $\omega_9 = \exp(2 \pi \mathrm{i}/9)$, the coefficient $a$ belongs to $GR(3^2,n)$ and is not a divisor of $0$. Furthermore, the indices of basis states are immersed in $GR(3^2,n)$ as $j = j_0 + j_1 x + \cdots$, where $j_0,j_1\cdots$ are indices of consecutive qutrits. Thus $j$ runs over all the elements from the Teichm\"uller set. Furthermore, due to characteristic $9$ of $GR(3^2,n)$ one has $(j+3r)^3 = j^3 + 9j^2 r + 27 jr^2 + 27 r^3 = j^3$, thus allowing to exchange $j$ for any other coset representant of $GR(3^2,n)/(3)$.

Indeed, by straightforward calculations we obtain

\begin{subequations}\label{eq:fid3n}
    \begin{align}
    & \abs{\ev{\vb{X^bZ^c}}{f_{\text{MUB},3^n}}}  = \frac{1}{3^n}\abs{\sum_{i,j = 0}^{3^n-1} \omega_9^{\tr_{GR}(ai^{3}) - \tr_{GR}(aj^{ 3})}\omega_3^{\sum_\nu i_\nu c_\nu}\ip{j-b}{i}}\\
    & \hspace{2 cm} = \frac{1}{3^n}\abs{\sum_{i = 0}^{3^n-1} \omega_9^{\tr_{GR}(a i^{3} - a (i+b)^{3})}\omega_3^{\sum_\nu i_\nu c_\nu}} = \frac{1}{3^n}\abs{\sum_{i = 0}^{3^n-1} \omega_9^{-3(\tr_{GR}[a b i^{ 2}  + a b^2 i] - \sum_\nu i_\nu c_\nu)}}.
\end{align}    
\end{subequations}

Notice that all the elements $i,b \in GR(3^2,n)$ have the polynomial coefficients in $\{0,1,2\}$, and so are the coefficients~$c_{\nu}$. Furthermore one may choose $\tilde{a}$ with coefficients from $\{0,1,2\}$ that $3a = 3\tilde{a}$. This, together with the fact that the selected polynomial is irreducible over $\fF_3$, implies that the multiplication and addition of all elements in the bracket in the exponent behaves as in $\fF_{3^n}$. 
More rigorously, we leverage the fact that the residue field $GR(3^2,n)/(3)$, where $(3)$ is maximal ideal generated by $3$, is isomorphic to $\fF_{3^n}$ and treat the elements $a,i,b$ as representants.
Finally, we may exchange the trace with respect to $GR(3^2,1) = \fZ/9\fZ$ with the trace over $\fF_{3^n}$.

Thus, we replace $\omega_9^3$ with $\omega_3$ and continue the computations from \eqref{eq:fid3n} as in the previous subsection

\begin{subequations}
\begin{align}
         &~~ \frac{1}{3^n}\abs{\sum_{i = 0}^{3^n-1} \omega_3^{-\tr_{GR}(a b i^{2} + a b^2 i - \sum_\nu i_\nu c_\nu)}}\ = \frac{1}{3^n}\abs{\sum_{i = 0}^{3^n-1} \omega_3^{-\tr_{\fF}\qty(abi^2 + (ab^2-c_*)i)}}  \\
         &= \frac{1}{3^n}\abs{\sum_{i = 0}^{3^n-1} \omega_3^{-\tr_{\fF}\qty[ab\qty(i + \frac{ab^2-c_*}{2ab})^2]}}  = \frac{1}{3^n}\abs{\sum_{j = 0}^{3^n-1} \omega_3^{-\tr_{\fF}\qty[ab j^2]}} \label{eq:failure_point}\\
         & = \frac{1-\delta_{b0}}{3^n}\prod_{k=1}^n \abs{\sum_{i_k=0}^{2}\omega_3^{\lambda_k i_k^2}} + \frac{\delta_{b0}}{3^n}\abs{\sum_{i = 0}^{3^n-1} \omega_3^{\tr_{\fF}\qty(c_*i)}} = (1-\delta_{b0})3^{-n/2} + \delta_{b0}\delta_{c0},
\end{align} 
\end{subequations}
where once again $c_*$ is the unique element in $\fF_{3^n}$ such that $c_*i = \sum_\nu i_\nu c_\nu$, and we twice relabelled our states, first with $j = i + \frac{ab^2-c_*}{2ab}$, and then by changing the basis to $\{i_k\}$ which diagonalize quadratic form $\tr_{\fF}[j^2]$, with eigenvalues $\lambda_k$.

On the other hand, we could assume that $a$ is a divisor of zero. But then the coeficients $3ab$ and $3ab^2$ from \eqref{eq:fid3n} would be equal zero, reducing the expression to the standard Fourier sum equal zero if and only if $c$ is nonzero. This immediately undermines the subsequent expressions necessary for fulfillment of MUB conditions. 

\section{Unbiased  bipartite Hadamard matrices 
of dimension $9$ and $25$
}\label{app:examples}

We provide below exemplary Butson-Hadamard matrices in dimensions $d = 9$ and $d = 25$, based on the constructions given in Theorems~\ref{thm:f_mub_3} and~\ref{thm:f_mub}, respectively, with $a = 1$.
These matrices are formed
of isoentangled vectors
belonging to the composite Hilbert spaces
of dimension $p^2$ with $p=3,5$,
as the purity of the partial traces
of all projectors is the same, given by \eqref{eq:isoMUHs_ent}.

\begin{equation}
{\tiny
    \begin{aligned}
        \LOG(H_9) = \frac{2\pi}{9} \text{i}\mqty(\bullet & \bullet & \bullet & \bullet & \bullet & \bullet & \bullet & \bullet & \bullet \\
 1 & 4 & 7 & 1 & 4 & 7 & 1 & 4 & 7 \\
 8 & 5 & 2 & 8 & 5 & 2 & 8 & 5 & 2 \\
 4 & 4 & 4 & 7 & 7 & 7 & 1 & 1 & 1 \\
 8 & 2 & 5 & 2 & 5 & 8 & 5 & 8 & 2 \\
 \bullet & 6 & 3 & 3 & \bullet & 6 & 6 & 3 & \bullet \\
 5 & 5 & 5 & 2 & 2 & 2 & 8 & 8 & 8 \\
 \bullet & 3 & 6 & 6 & \bullet & 3 & 3 & 6 & \bullet \\
 1 & 7 & 4 & 7 & 4 & 1 & 4 & 1 & 7),
 &&
\LOG(H_{25}) =  \frac{2\pi}{5}\text{i}\left(
\begin{array}{ccccccccccccccccccccccccc}
 \bullet & \bullet & \bullet & \bullet & \bullet & \bullet & \bullet & \bullet & \bullet & \bullet & \bullet & \bullet & \bullet & \bullet & \bullet & \bullet & \bullet & \bullet & \bullet & \bullet & \bullet & \bullet & \bullet & \bullet & \bullet \\
 4 & \bullet & 1 & 2 & 3 & 4 & \bullet & 1 & 2 & 3 & 4 & \bullet & 1 & 2 & 3 & 4 & \bullet & 1 & 2 & 3 & 4 & \bullet & 1 & 2 & 3 \\
 2 & 4 & 1 & 3 & \bullet & 2 & 4 & 1 & 3 & \bullet & 2 & 4 & 1 & 3 & \bullet & 2 & 4 & 1 & 3 & \bullet & 2 & 4 & 1 & 3 & \bullet \\
 3 & 1 & 4 & 2 & \bullet & 3 & 1 & 4 & 2 & \bullet & 3 & 1 & 4 & 2 & \bullet & 3 & 1 & 4 & 2 & \bullet & 3 & 1 & 4 & 2 & \bullet \\
 1 & \bullet & 4 & 3 & 2 & 1 & \bullet & 4 & 3 & 2 & 1 & \bullet & 4 & 3 & 2 & 1 & \bullet & 4 & 3 & 2 & 1 & \bullet & 4 & 3 & 2 \\
 \bullet & \bullet & \bullet & \bullet & \bullet & 1 & 1 & 1 & 1 & 1 & 2 & 2 & 2 & 2 & 2 & 3 & 3 & 3 & 3 & 3 & 4 & 4 & 4 & 4 & 4 \\
 2 & 3 & 4 & \bullet & 1 & 3 & 4 & \bullet & 1 & 2 & 4 & \bullet & 1 & 2 & 3 & \bullet & 1 & 2 & 3 & 4 & 1 & 2 & 3 & 4 & \bullet \\
 \bullet & 2 & 4 & 1 & 3 & 1 & 3 & \bullet & 2 & 4 & 2 & 4 & 1 & 3 & \bullet & 3 & \bullet & 2 & 4 & 1 & 4 & 1 & 3 & \bullet & 2 \\
 3 & 1 & 4 & 2 & \bullet & 4 & 2 & \bullet & 3 & 1 & \bullet & 3 & 1 & 4 & 2 & 1 & 4 & 2 & \bullet & 3 & 2 & \bullet & 3 & 1 & 4 \\
 \bullet & 4 & 3 & 2 & 1 & 1 & \bullet & 4 & 3 & 2 & 2 & 1 & \bullet & 4 & 3 & 3 & 2 & 1 & \bullet & 4 & 4 & 3 & 2 & 1 & \bullet \\
 \bullet & \bullet & \bullet & \bullet & \bullet & 2 & 2 & 2 & 2 & 2 & 4 & 4 & 4 & 4 & 4 & 1 & 1 & 1 & 1 & 1 & 3 & 3 & 3 & 3 & 3 \\
 4 & \bullet & 1 & 2 & 3 & 1 & 2 & 3 & 4 & \bullet & 3 & 4 & \bullet & 1 & 2 & \bullet & 1 & 2 & 3 & 4 & 2 & 3 & 4 & \bullet & 1 \\
 1 & 3 & \bullet & 2 & 4 & 3 & \bullet & 2 & 4 & 1 & \bullet & 2 & 4 & 1 & 3 & 2 & 4 & 1 & 3 & \bullet & 4 & 1 & 3 & \bullet & 2 \\
 \bullet & 3 & 1 & 4 & 2 & 2 & \bullet & 3 & 1 & 4 & 4 & 2 & \bullet & 3 & 1 & 1 & 4 & 2 & \bullet & 3 & 3 & 1 & 4 & 2 & \bullet \\
 \bullet & 4 & 3 & 2 & 1 & 2 & 1 & \bullet & 4 & 3 & 4 & 3 & 2 & 1 & \bullet & 1 & \bullet & 4 & 3 & 2 & 3 & 2 & 1 & \bullet & 4 \\
 \bullet & \bullet & \bullet & \bullet & \bullet & 3 & 3 & 3 & 3 & 3 & 1 & 1 & 1 & 1 & 1 & 4 & 4 & 4 & 4 & 4 & 2 & 2 & 2 & 2 & 2 \\
 \bullet & 1 & 2 & 3 & 4 & 3 & 4 & \bullet & 1 & 2 & 1 & 2 & 3 & 4 & \bullet & 4 & \bullet & 1 & 2 & 3 & 2 & 3 & 4 & \bullet & 1 \\
 \bullet & 2 & 4 & 1 & 3 & 3 & \bullet & 2 & 4 & 1 & 1 & 3 & \bullet & 2 & 4 & 4 & 1 & 3 & \bullet & 2 & 2 & 4 & 1 & 3 & \bullet \\
 4 & 2 & \bullet & 3 & 1 & 2 & \bullet & 3 & 1 & 4 & \bullet & 3 & 1 & 4 & 2 & 3 & 1 & 4 & 2 & \bullet & 1 & 4 & 2 & \bullet & 3 \\
 1 & \bullet & 4 & 3 & 2 & 4 & 3 & 2 & 1 & \bullet & 2 & 1 & \bullet & 4 & 3 & \bullet & 4 & 3 & 2 & 1 & 3 & 2 & 1 & \bullet & 4 \\
 \bullet & \bullet & \bullet & \bullet & \bullet & 4 & 4 & 4 & 4 & 4 & 3 & 3 & 3 & 3 & 3 & 2 & 2 & 2 & 2 & 2 & 1 & 1 & 1 & 1 & 1 \\
 \bullet & 1 & 2 & 3 & 4 & 4 & \bullet & 1 & 2 & 3 & 3 & 4 & \bullet & 1 & 2 & 2 & 3 & 4 & \bullet & 1 & 1 & 2 & 3 & 4 & \bullet \\
 2 & 4 & 1 & 3 & \bullet & 1 & 3 & \bullet & 2 & 4 & \bullet & 2 & 4 & 1 & 3 & 4 & 1 & 3 & \bullet & 2 & 3 & \bullet & 2 & 4 & 1 \\
 \bullet & 3 & 1 & 4 & 2 & 4 & 2 & \bullet & 3 & 1 & 3 & 1 & 4 & 2 & \bullet & 2 & \bullet & 3 & 1 & 4 & 1 & 4 & 2 & \bullet & 3 \\
 3 & 2 & 1 & \bullet & 4 & 2 & 1 & \bullet & 4 & 3 & 1 & \bullet & 4 & 3 & 2 & \bullet & 4 & 3 & 2 & 1 & 4 & 3 & 2 & 1 & \bullet \\
\end{array}
\right) 
    \end{aligned}}
\end{equation}
Following the notation of~\cite{TZ05},
$\LOG$ denotes here element-wise logarithm and $\bullet$ represents zero entries.
The first column of both matrices corresponds to the appropriate fiducial vector $\ket{f_{\text{MUB}}}$, with the remaining columns obtained by local action of phase operators $Z^i\otimes Z^j$ with $i, j = 1, \hdots, p$. 
The other unbiased Hadamard matrices, comprising the complete set of mutually unbiased bases,
can be obtained by local action of products of shift operators $X^k\otimes X^l$ with $k, l = 1,\hdots,p$. 

Note that $H_9$ belongs to the Butson class $BH(9,9)$, while $H_{25}$ forms an example of the Butson class $BH(25,5)$. At the same time, both Hadamard matrices are equivalent, in the standard Hadamard-related sense, to the tensor product of local Fourier matrices $F_p^{\otimes 2}$; the equivalence is a consequence of being generated from a single fiducial vector by the action of $Z^i\otimes Z^j$.

Similar remarks extend to $n>2$ and $p\geq 3$, with Butson classes $BH(3^n,9)$ and $BH(p^n,p)$ covering all the presented solutions and equivalence to tensor powers of local Fourier matrices $F_p^{\otimes n}$ holding across the board and extending to the sporadic cases for $d=2, 4$ and $d = 9$ presented in Appendix~\ref{app:sporadic_set}.

\section{Sporadic set of isoentangled Hadamard matrices in $d=9$}
\label{app:sporadic_set}

In dimension $d=9$ we were able to find a complete set of mutually unbiased isoentangled Hadamard matrices which do not originate from a fiducial vector relative to the product Weyl–Heisenberg group. Interestingly, they can still be reframed as three independent orbits of a subgroup of the product WH group, each with its own fiducial vector. The fiducial vectors are of the form
\begin{equation}
    \operatorname{LOG}\qty(\ket{\psi_0}) = \frac{2\pi}{3}\text{i}
    \begin{pmatrix}
    0\\1\\1\\1\\1\\0\\1\\1\\0
    \end{pmatrix}, \quad
    \operatorname{LOG}\qty(\ket{\psi_1}) = \frac{2\pi}{3}\text{i}
    \begin{pmatrix}
    0\\0\\1\\0\\0\\1\\0\\1\\0
    \end{pmatrix}, \quad
    \operatorname{LOG}\qty(\ket{\psi_2}) = \frac{2\pi}{3}\text{i}
    \begin{pmatrix}
    0\\0\\0\\0\\0\\0\\1\\2\\0
    \end{pmatrix}, 
\end{equation}
where $\operatorname{LOG}(\cdot)$ should be understood as element-wise logarithm of the complex vectors. It is thus apparent that this solution is of Butson type $BH(9, \,3)$, only involving third roots of unity. Explicit computation shows that the purity of their reductions is given by $\gamma(\ket{\psi_i}) = 5/9$, in agreement with \eqref{eq:isoMUHs_ent}, necessary to reproduce a full set of isoentangled mutually unbiased Hadamard matrices. Using the complete product WH group one can generate three complete families of MUHs $\mathcal{B}_i$ indexed by the power of the $X$ operator on the first subsystem

\begin{equation}
    \mathcal{B}_{i} = \qty{H_{3l+j}^{(i)} = \qty{(X^j\otimes X^{i})Z_{\vb{k}}\ket{\psi_{l}}}_{\vb{k}=0}^{8}}_{j,l=0}^{2},
\end{equation}
where $\vb{k}$ indexes elements within the bases and $j, l$ index jointly different bases of the MUBs. It is interesting to note that bases from different families show a product structure reminiscent of unbiasedness, with $\abs{\ip{\psi}{\phi}}^2 \in \qty{0,1/9,1/3}$ for $\ket{\psi}\in\mathcal{B}_i, \ket{\phi}\in\mathcal{B}_j, i\neq j$.

To provide an explicit example, the first set $\mathcal{B}_0$ can be written in a succinct manner, after dephasing of the first component for each vector, as

\allowdisplaybreaks
\begin{equation}
\hspace{-0.5 cm}
{\tiny 
\setlength{\arraycolsep}{3.3pt} 
\centering
\begin{aligned}
& \operatorname{LOG}\qty(H^{(0)}_0)=\frac{2\pi}{3}\text{i}
\begin{pmatrix}
 \bullet & \bullet & \bullet & \bullet & \bullet & \bullet & \bullet & \bullet & \
\bullet \\
 1 & 2 & \bullet & 1 & 2 & \bullet & 1 & 2 & \bullet \\
 1 & \bullet & 2 & 1 & \bullet & 2 & 1 & \bullet & 2 \\
 1 & 1 & 1 & 2 & 2 & 2 & \bullet & \bullet & \bullet \\
 1 & 2 & \bullet & 2 & \bullet & 1 & \bullet & 1 & 2 \\
 \bullet & 2 & 1 & 1 & \bullet & 2 & 2 & 1 & \bullet \\
 1 & 1 & 1 & \bullet & \bullet & \bullet & 2 & 2 & 2 \\
 1 & 2 & \bullet & \bullet & 1 & 2 & 2 & \bullet & 1 \\
 \bullet & 2 & 1 & 2 & 1 & \bullet & 1 & \bullet & 2
\end{pmatrix}, \
\operatorname{LOG}\qty(H^{(0)}_1)=\frac{2\pi}{3}\text{i}
\begin{pmatrix}
 \bullet & \bullet & \bullet & \bullet & \bullet & \bullet & \bullet & \bullet & \
\bullet \\
 \bullet & 1 & 2 & \bullet & 1 & 2 & \bullet & 1 & 2 \\
 2 & 1 & \bullet & 2 & 1 & \bullet & 2 & 1 & \bullet \\
 \bullet & \bullet & \bullet & 1 & 1 & 1 & 2 & 2 & 2 \\
 \bullet & 1 & 2 & 1 & 2 & \bullet & 2 & \bullet & 1 \\
 2 & 1 & \bullet & \bullet & 2 & 1 & 1 & \bullet & 2 \\
 2 & 2 & 2 & 1 & 1 & 1 & \bullet & \bullet & \bullet \\
 \bullet & 1 & 2 & 2 & \bullet & 1 & 1 & 2 & \bullet \\
 \bullet & 2 & 1 & 2 & 1 & \bullet & 1 & \bullet & 2
\end{pmatrix}, \
\operatorname{LOG}\qty(H^{(0)}_2)=\frac{2\pi}{3}\text{i}
\begin{pmatrix}
 \bullet & \bullet & \bullet & \bullet & \bullet & \bullet & \bullet & \bullet & \
\bullet \\
 \bullet & 1 & 2 & \bullet & 1 & 2 & \bullet & 1 & 2 \\
 2 & 1 & \bullet & 2 & 1 & \bullet & 2 & 1 & \bullet \\
 2 & 2 & 2 & \bullet & \bullet & \bullet & 1 & 1 & 1 \\
 \bullet & 1 & 2 & 1 & 2 & \bullet & 2 & \bullet & 1 \\
 \bullet & 2 & 1 & 1 & \bullet & 2 & 2 & 1 & \bullet \\
 \bullet & \bullet & \bullet & 2 & 2 & 2 & 1 & 1 & 1 \\
 \bullet & 1 & 2 & 2 & \bullet & 1 & 1 & 2 & \bullet \\
 2 & 1 & \bullet & 1 & \bullet & 2 & \bullet & 2 & 1
\end{pmatrix}, \\[6mm]
& \operatorname{LOG}\qty(H^{(0)}_3)=\frac{2\pi}{3}\text{i}
\begin{pmatrix}
 \bullet & \bullet & \bullet & \bullet & \bullet & \bullet & \bullet & \bullet & \
\bullet \\
 \bullet & 1 & 2 & \bullet & 1 & 2 & \bullet & 1 & 2 \\
 1 & \bullet & 2 & 1 & \bullet & 2 & 1 & \bullet & 2 \\
 \bullet & \bullet & \bullet & 1 & 1 & 1 & 2 & 2 & 2 \\
 \bullet & 1 & 2 & 1 & 2 & \bullet & 2 & \bullet & 1 \\
 1 & \bullet & 2 & 2 & 1 & \bullet & \bullet & 2 & 1 \\
 \bullet & \bullet & \bullet & 2 & 2 & 2 & 1 & 1 & 1 \\
 1 & 2 & \bullet & \bullet & 1 & 2 & 2 & \bullet & 1 \\
 \bullet & 2 & 1 & 2 & 1 & \bullet & 1 & \bullet & 2
\end{pmatrix}, \
\operatorname{LOG}\qty(H^{(0)}_4)=\frac{2\pi}{3}\text{i}
\begin{pmatrix}
\bullet & \bullet & \bullet & \bullet & \bullet & \bullet & \bullet & \bullet & \bullet \\
\bullet & 1 & 2 & \bullet & 1 & 2 & \bullet & 1 & 2 \\
1 & \bullet & 2 & 1 & \bullet & 2 & 1 & \bullet & 2 \\
\bullet & \bullet & \bullet & 1 & 1 & 1 & 2 & 2 & 2 \\
1 & 2 & \bullet & 2 & \bullet & 1 & \bullet & 1 & 2 \\
\bullet & 2 & 1 & 1 & \bullet & 2 & 2 & 1 & \bullet \\
\bullet & \bullet & \bullet & 2 & 2 & 2 & 1 & 1 & 1 \\
\bullet & 1 & 2 & 2 & \bullet & 1 & 1 & 2 & \bullet \\
1 & \bullet & 2 & \bullet & 2 & 1 & 2 & 1 & \bullet
\end{pmatrix}, \
\operatorname{LOG}\qty(H^{(0)}_5)=\frac{2\pi}{3}\text{i}
\begin{pmatrix}
\bullet & \bullet & \bullet & \bullet & \bullet & \bullet & \bullet & \bullet & \bullet \\
1 & 2 & \bullet & 1 & 2 & \bullet & 1 & 2 & \bullet \\
\bullet & 2 & 1 & \bullet & 2 & 1 & \bullet & 2 & 1 \\
\bullet & \bullet & \bullet & 1 & 1 & 1 & 2 & 2 & 2 \\
\bullet & 1 & 2 & 1 & 2 & \bullet & 2 & \bullet & 1 \\
1 & \bullet & 2 & 2 & 1 & \bullet & \bullet & 2 & 1 \\
\bullet & \bullet & \bullet & 2 & 2 & 2 & 1 & 1 & 1 \\
\bullet & 1 & 2 & 2 & \bullet & 1 & 1 & 2 & \bullet \\
1 & \bullet & 2 & \bullet & 2 & 1 & 2 & 1 & \bullet
\end{pmatrix}, \\[6mm]
& \operatorname{LOG}\qty(H^{(0)}_6)=\frac{2\pi}{3}\text{i}
\begin{pmatrix}
 \bullet & \bullet & \bullet & \bullet & \bullet & \bullet & \bullet & \bullet & \
\bullet \\
 \bullet & 1 & 2 & \bullet & 1 & 2 & \bullet & 1 & 2 \\
 \bullet & 2 & 1 & \bullet & 2 & 1 & \bullet & 2 & 1 \\
 \bullet & \bullet & \bullet & 1 & 1 & 1 & 2 & 2 & 2 \\
 \bullet & 1 & 2 & 1 & 2 & \bullet & 2 & \bullet & 1 \\
 \bullet & 2 & 1 & 1 & \bullet & 2 & 2 & 1 & \bullet \\
 1 & 1 & 1 & \bullet & \bullet & \bullet & 2 & 2 & 2 \\
 2 & \bullet & 1 & 1 & 2 & \bullet & \bullet & 1 & 2 \\
 \bullet & 2 & 1 & 2 & 1 & \bullet & 1 & \bullet & 2
\end{pmatrix}, \
\operatorname{LOG}\qty(H^{(0)}_7)=\frac{2\pi}{3}\text{i}
\begin{pmatrix}
 \bullet & \bullet & \bullet & \bullet & \bullet & \bullet & \bullet & \bullet & \
\bullet \\
 \bullet & 1 & 2 & \bullet & 1 & 2 & \bullet & 1 & 2 \\
 \bullet & 2 & 1 & \bullet & 2 & 1 & \bullet & 2 & 1 \\
 1 & 1 & 1 & 2 & 2 & 2 & \bullet & \bullet & \bullet \\
 2 & \bullet & 1 & \bullet & 1 & 2 & 1 & 2 & \bullet \\
 \bullet & 2 & 1 & 1 & \bullet & 2 & 2 & 1 & \bullet \\
 \bullet & \bullet & \bullet & 2 & 2 & 2 & 1 & 1 & 1 \\
 \bullet & 1 & 2 & 2 & \bullet & 1 & 1 & 2 & \bullet \\
 \bullet & 2 & 1 & 2 & 1 & \bullet & 1 & \bullet & 2
\end{pmatrix}, \
\operatorname{LOG}\qty(H^{(0)}_8)=\frac{2\pi}{3}\text{i}
\begin{pmatrix}
 \bullet & \bullet & \bullet & \bullet & \bullet & \bullet & \bullet & \bullet & \
\bullet \\
 1 & 2 & \bullet & 1 & 2 & \bullet & 1 & 2 & \bullet \\
 2 & 1 & \bullet & 2 & 1 & \bullet & 2 & 1 & \bullet \\
 2 & 2 & 2 & \bullet & \bullet & \bullet & 1 & 1 & 1 \\
 2 & \bullet & 1 & \bullet & 1 & 2 & 1 & 2 & \bullet \\
 2 & 1 & \bullet & \bullet & 2 & 1 & 1 & \bullet & 2 \\
 2 & 2 & 2 & 1 & 1 & 1 & \bullet & \bullet & \bullet \\
 2 & \bullet & 1 & 1 & 2 & \bullet & \bullet & 1 & 2 \\
 2 & 1 & \bullet & 1 & \bullet & 2 & \bullet & 2 & 1
\end{pmatrix},
\end{aligned}}
\label{rogue_MUB}
\end{equation}
where 
$\bullet$ denotes entries whose phases are zeros. 

The tri-fiducial structure of the solution is a straightforward consequence of the Ansatz used to obtain it. We considered three vectors $|\psi_0\ra$, $\ket{\psi_1}$, and $|\psi_2\ra$ with equal entanglement $\gamma(|\psi_i\ra)=5/9$, and product magic $M(|\psi_i\ra)=\text{const}$, which belonged to different orbits for the product WH group $\mathcal{O}_{|\psi_0\ra}\cap\mathcal{O}_{|\psi_1\ra}\cap\mathcal{O}_{|\psi_2\ra}=\emptyset$. Since any MUB can be represented as a Hadamard matrix~\cite{mcnulty2024mutually}, we adopted the following ansatz for these vectors
\begin{equation}
|\psi_i\rangle = \frac{1}{3}
\begin{pmatrix}
1, \omega_3^{i_1}, \omega_3^{i_2}, \dots, \omega_3^{i_{8}}
\end{pmatrix},
\label{ansatz}
\end{equation}
where $\omega_3=e^{2\pi \mathrm{i}/3}$ and $(i_1,\ldots,i_8)\in\mathbb{Z}_3^{\times 8}$. 

Due to a limited size of the problem an exhaustive search over $\mathbb{Z}_3^{\times 8}$ has been performed, which resulted
in a triplet
that satisfied the aforementioned constraints. 

The structure illustrated above presents itself as a further extension of the idea of fiducial sets. Namely, the basic idea of a single fiducial vector present in SIC was first extended to a pair of fiducial vectors (including a single vector from the computational basis) to produce a complete set of $d+1$ MUBs in odd prime-power dimensions. Now it is subsequently extended to the sporadic solution presented here, consisting in a triplet of fiducial vectors producing a complete set of unbiased Hadamard matrices
under a subgroup of the product WH group, which can additionally be complemented by the computational basis. A similar idea of a single basis generating complex projective 2-designs under $U(1)$ action has been introduced recently in~\cite{avella2025cyclic}. Altogether, they present a generalized concept of fiducial structures, or group frames as conceived in~\cite{Feng2024}, and as such 
merit deeper investigation.

\bibliography{references.bib}

\end{document}